\documentclass{rQUF2e-mod}

\usepackage{amsbsy}
\usepackage{amsfonts}
\usepackage{amssymb}
\usepackage{amsmath}
\usepackage{natbib}
\usepackage{color}
\usepackage[normalem]{ulem}
\usepackage{verbatim}

\renewcommand{\P}{\mathbb{P}}
\newcommand{\Q}{\mathbb{Q}}
\newcommand{\R}{\mathbb{R}}

\newcommand{\W}{\mathcal{W}}

\newcommand{\lrparen}[1]{\left(#1\right)}

\newcommand{\lrsquare}[1]{\left[#1\right]}
\newcommand{\lsquare}[1]{\left[#1\right.}
\newcommand{\rsquare}[1]{\left.#1\right]}
\newcommand{\lrcurly}[1]{\left\{#1\right\}}
\newcommand{\lcurly}[1]{\left\{#1\right.}
\newcommand{\rcurly}[1]{\left.#1\right\}}

\newcommand{\Ft}[1]{\mathcal{F}_{#1}}
\newcommand{\Lnp}[3]{L_{#1}^{#2}(#3)}

\newcommand{\Ldp}[2]{\Lnp{d}{#1}{#2}}

\newcommand{\LdpF}[1]{\Ldp{p}{\Ft{#1}}}
\newcommand{\LdqF}[1]{\Ldp{q}{\Ft{#1}}}

\newcommand{\Ldi}[1]{\Ldp{\infty}{#1}}
\newcommand{\LdiF}[1]{\Ldi{\Ft{#1}}}
\newcommand{\Ldo}[1]{\Ldp{1}{#1}}
\newcommand{\LdoF}[1]{\Ldo{\Ft{#1}}}
\newcommand{\Ldz}[1]{\Ldp{0}{#1}}
\newcommand{\LdzF}[1]{\Ldz{\Ft{#1}}}
\newcommand{\LdpK}[3]{\Ldp{#1}{\Ft{#2};#3}}

\newcommand{\Lp}[2]{\Lnp{}{#1}{#2}}
\newcommand{\LpF}[1]{\Lp{p}{\Ft{#1}}}

\newcommand{\LiF}[1]{\Lp{\infty}{\Ft{#1}}}
\newcommand{\LoF}[1]{\Lp{1}{\Ft{#1}}}

\newcommand{\E}[1]{\mathbb{E}\lrsquare{#1}}
\newcommand{\EQ}[1]{\mathbb{E}^{\mathbb{Q}}\lrsquare{#1}}

\newcommand{\Et}[2]{\E{\left.#1 \right| \mathcal{F}_{#2}}}
\newcommand{\EQt}[2]{\EQ{\left.#1 \right| \mathcal{F}_{#2}}}

\newcommand{\dQdP}{\frac{d\mathbb{Q}}{d\mathbb{P}}}
\newcommand{\dQndP}[1]{\frac{d\mathbb{Q}_{#1}}{d\mathbb{P}}}
\newcommand{\dQidP}{\dQndP{i}}

\newcommand{\genseq}[4]{\lrparen{#1_#4}_{#4=#2}^{#3}}
\newcommand{\seq}[1]{\genseq{#1}{0}{T}{t}}

\newcommand{\FYvblank}{F_{(Y,v)}^{M_t}}
\newcommand{\FYv}[1]{\FYvblank\lrsquare{#1}}

\newcommand{\FQwblank}{\tilde{F}_{(\mathbb{Q},w)}^{M_t}}
\newcommand{\FQw}[1]{\tilde{F}_{(\mathbb{Q},w)}^{M_t}\lrsquare{#1}}

\newcommand{\trans}[1]{#1^{\text{T}}}
\newcommand{\transp}[1]{\trans{\lrparen{#1}}}
\newcommand{\prp}[1]{#1^{\perp}}

\newcommand{\plus}[1]{#1^+}
\newcommand{\plusp}[1]{\plus{(#1)}}

\newcommand{\diag}[1]{\operatorname{diag}\lrparen{#1}}
\newcommand{\cl}{\operatorname{cl}}
\newcommand{\co}{\operatorname{co}}

\newcommand{\Int}{{\rm int\,}}

\newcommand{\as}{\text{a.s.}}

\newcommand{\st}{\text{ s.t. }}

\newcommand{\1}{\mathbf{1}}

\DeclareMathOperator{\esssup}{ess\,sup}

\title{Time consistency of dynamic risk measures in markets with transaction costs}
\author{Zachary Feinstein$\dag$ and Birgit
    Rudloff$^*$$\dag$\thanks{$^*$Corresponding author. Email: brudloff@princeton.edu}\\
    \vspace{12pt}  \normalfont{$\dag$Princeton University, Department of Operations Research and Financial Engineering; and Bendheim Center for Finance, Princeton, NJ 08544, USA}}
\begin{document}
\maketitle

\begin{abstract}
Set-valued dynamic risk measures are defined on $\LdpF{T}$ with $0 \leq p \leq \infty$ and with an image space in the power set of $\LdpF{t}$. Primal and dual representations of dynamic risk measures are deduced.  Definitions of different time consistency properties in the set-valued framework are given.  It is shown that the recursive form for multivariate risk measures as well as an additive property for the acceptance sets is equivalent to a stronger time consistency property called multi-portfolio time consistency.
\end{abstract}

\section{Introduction}
Scalar risk measures, defined axiomatically by~\cite{AD99} and further studied in~\cite{D02,FS02,FS11}, provide the minimal capital required (in some num\'{e}raire) needed to cover the risk of a contingent claim.
In markets with transaction costs multivariate risks, that is $X\in\LdpF{T}$, need to be considered if one does not want to assume that $X$ needs to be liquidated into a given num\'{e}raire asset first, or if one does not want to choose one particular asset as a num\'{e}raire. In such a num\'{e}raire free model initiated by \cite{K99} the value of the risk measure is a collection of portfolio \emph{vectors} describing, for example, the amount of money in different currencies that compensate for the risk of $X$. Then, `minimal capital requirements' correspond to the set of efficient points (the efficient frontier) of the set of risk compensating portfolio vectors. Mathematically it is easier to work with the set of all risk compensating portfolio vectors than with the set of its minimal elements (nevertheless, in both cases the value of the risk measure would be a set!). For example, the whole set is convex for convex risk measures, whereas the set of minimal elements is rarely a convex set. Thus, considering multivariate risks in a market with transaction costs leads naturally to set-valued risk measures.
In a static, one-period framework set-valued risk measures and their dual representations have been studied in \cite{JMT04,HR08,HH10,HHR10} and their values and minimal elements have been computed in \cite{LR11, HRY12}.

In this paper we want to explore set-valued risk measures in a dynamic set-up, when the risk of $X$ is reevaluated at discrete time points $t\in\{0,1,\dots,T\}$. We will define dynamic set-valued risk measures and study their dual representation. Most importantly, we will study time consistency.
In the scalar case, those results are well understood. We refer to
\cite{R04,BN04,DS05,RS05,CDK06,FP06,AD07,CS09,CK10,AP10,FS11}
for the discrete time case and \cite{FG04,D06,DPRG10}
for scalar dynamic risk measures in continuous time.  Time consistency means that risk relations between portfolios are constant (backwards) in time.  It is a well known result in the scalar case that time consistency is equivalent to a recursive representation which allows for calculations via Bellman's principle. Is something similar true in the set-valued case?

This question will be answered in section~\ref{sec_time_consistent}. The difficulty lies in the fact that the time consistency property for scalar risk measures can be generalized to set-valued risk measures in different ways. The most intuitive generalization we will call time consistency. We will show that the equivalence between a recursive form of the risk measure and time consistency, which is a central result in the scalar case, does not hold in the set-valued framework.  Instead, we propose an alternative generalization, which we will call multi-portfolio time consistency and show that this property is indeed equivalent to the recursive form as well as to an additive property for the acceptance sets. Multi-portfolio time consistency is a stronger property than time consistency, however in the scalar framework both notions coincide. In section~\ref{subsec composition}, the idea described in the scalar case in \cite{CS09} to generate a new time-consistent risk measure from any dynamic risk measure by composing them backwards in time is studied in the set-valued case.
Section~\ref{sec_dual_rep} studies the dual representation of convex and coherent dynamic set-valued risk measures using the set-valued duality approach by~\cite{H09}.  Surprisingly, the conditional expectations we would expect in the dual representation of dynamic risk measures (having the scalar case in mind) appear only after a transformation of dual variables from the classical dual pairs in the set-valued theory of \cite{H09} to dual pairs involving vector probability measures. This leads to dual representation in the spirit of the well known scalar case (theorem~\ref{thm_dual} and corollary~\ref{cor_dual}).
The paper concludes with section~\ref{sec_examples} discussing two examples of dynamic risk measures.  First, we show that the set of superhedging portfolios in markets with proportional transaction costs is a multi-portfolio time consistent dynamic coherent risk measure. We show that the recursive form leads to and simplifies the proof of the recursive algorithm given in \cite{LR11}. This result gives a hint that the set-valued recursive form of multi-portfolio time consistent risk measures is very useful in practice and might lead to a set-valued analog of Bellman's principle.  The second example is the dynamic set-valued version of the average value at risk, which is not multi-portfolio time consistent, but a multi-portfolio time consistent version can be obtained by backward composition.

Dynamic risk measures for markets with transaction costs appear in various settings in the very recent literature. In~\cite{BC12}, dynamic coherent risk measures are considered with transaction costs, accomplishing this by considering a family of scalar risk measures via dynamic coherent acceptability indices (studied in~\cite{BCZ11}). In~\cite{JB08} scalar dynamic coherent risk measures with respect to a set of numeraires are considered and time consistency is studied. In~\cite{TL12} dynamic set-valued coherent risk measures are defined for bounded (w.r.t. the solvency cone) random vectors. They provide a corresponding duality result w.r.t  measurable selectors, which does not involve conditional expectation representations in general. Time-consistency issues are not discussed in \cite{TL12}. Compared to these approaches in the literature, we can treat dynamic convex risk measures as well, we are able to provide a dual representation with a clear relation to the conditional expectation representation of the scalar case. Furthermore, the connection to linear vector optimization enables the computation of (dynamic) set-valued risk measures, see the two examples discussed in \cite{LR11, HRY12}.


\section{Dynamic risk measures}
\label{sec_dynamic_defn}
Consider a filtered probability space $(\Omega,\Ft{T},\seq{\mathcal{F}},\mathbb{P})$ satisfying the usual conditions with $\Ft{0}$ the trivial sigma algebra.  Let $|\cdot|$ be an arbitrary norm in $\mathbb{R}^d$. Denote by $\LdpF{t} = \Ldp{p}{\Omega,\Ft{t},\mathbb{P}}$, $p \in \lrsquare{0,\infty}$.  If $p = 0$ then $\LdzF{t}$ is the linear space of the equivalence classes of $\Ft{t}$-measurable functions $X: \Omega \to \mathbb{R}^d$.  For $p > 0$, $\LdpF{t}$ denotes the linear space of $\Ft{t}$-measurable functions $X: \Omega \to \mathbb{R}^d$ such that $||X||_p^p = \int_{\Omega}|X(\omega)|^p d\mathbb{P} < +\infty$ for $p \in (0,\infty)$, and $||X||_{\infty} = \esssup_{\omega \in \Omega} |X(\omega)| < +\infty$ for $p = \infty$.
We write $\LdpF{t}_+ = \lrcurly{X \in \LdpF{t}: \; X \in \mathbb{R}^d_+ \; \mathbb{P}-\as}$ for the closed convex cone of $\mathbb{R}^d$-valued $\Ft{t}$-measurable random vectors with non-negative components.  Note that an element $X \in \LdpF{t}$ has components $X_1,\dots,X_d$ in $\LpF{t}=L^p_1(\mathcal F_{t})$. We denote by $\LdpK{p}{t}{D_t}$ those random vectors in $\LdpF{t}$ that take $\P$-a.s. values in $D_t$.
(In-)equalities between random vectors
are always understood componentwise
in the $\mathbb{P}$-a.s. sense.
The multiplication between a random variable $\lambda\in \LiF{t}$ and a set of random vectors $D\subseteq\LdpF{T}$ is understood as $\lambda D=\{\lambda Y:Y\in D\}\subseteq\LdpF{T}$ with $(\lambda Y)(\omega)=\lambda(\omega)Y(\omega)$.

Let $d$ be the number of assets in a discrete time market with $t\in\{0,1,\dots,T\}$. As in \cite{K99} and discussed in \cite{S04,KS09}, the portfolios in this paper are in ``physical units" of an asset rather than the value in a fixed num\'{e}raire.  That is, for a portfolio $X \in \LdpF{t}$, the values of $X_i$ (for $1 \leq i \leq d$) are the number of units of asset $i$ in the portfolio at time $t$.

We will define set-valued risk measures as a function mapping a contingent claim into a set of portfolios which can be used to cover the risk of that claim.  In this way we must introduce the set of potential risk covering portfolios.  We will let $\tilde{M}_t$ be an $\Ft{t}$-measurable set such that $\tilde{M}_t(\omega)$ is a subspace of $\R^d$ for almost every $\omega \in \Omega$.  Then $M_t := \LdpK{p}{t}{\tilde{M}_t}$ is a closed (weak* closed if $p = +\infty$) linear subspace of $\LdpF{t}$, see section~5.4 and proposition~5.5.1 in \cite{KS09}.  We call the portfolios in $M_t$ the eligible portfolios.
The set of eligible portfolios includes possibilities such as: a single asset or currency, a basket of assets, or a specified ratio of assets.
\begin{example}
\label{ex_eligible}
In the scalar framework (see e.g. \cite{AD07,R04,DS05,CDK06,RS05,BN04,FP06,CS09,CK10}) the risk is covered by the num\'{e}raire only.  This corresponds to $\tilde{M}_t(\omega) = \lrcurly{m \in \R^d: \forall j \in \{2,3,\dots,d\}: m_j = 0}$ for almost every $\omega \in \Omega$; thus the set of eligible assets is \[M_t = \lrcurly{m \in \LdpF{t}: \forall j \in \{2,3,\dots,d\}: m_j = 0}\]
when the num\'{e}raire is taken to be the first asset.
\end{example}
It is possible that the set of eligible portfolios may change through time and even depend on the state of the world, however a more common situation is for a constant underlying set to be used.  That is, $\tilde{M}_t = M_0$ almost surely for every time $t \in \{1,\dots,T\}$ where $M_0$ is a subspace of $\R^d$.
This is the case for the eligible portfolios for scalar risk measures, as described in example~\ref{ex_eligible}.  Notice that a constant set of eligible assets in particular implies $M_t \subseteq M_{t+1}$ for all times $t$,
which we will see in section~\ref{sec_time_consistent} to be a necessary property for the equivalence of the recursive form and a property of the acceptance sets.

We will denote the set of nonnegative eligible portfolios by $\lrparen{M_t}_+= M_t \cap \LdpF{t}_+$ and the set of strictly positive eligible portfolios by $\lrparen{M_t}_{++}= M_t \cap \LdpF{t}_{++}$. We can analogously define $\lrparen{M_t}_-$ and $\lrparen{M_t}_{--}$.  We will assume $(M_t)_+$ is nontrivial, i.e. $(M_t)_+ \neq \{0\}$.  This property is satisfied for scalar risk measures, as given in example~\ref{ex_eligible}.

\subsection{Dynamic risk measures}
\label{subsec dynRM}
A conditional risk measure $R_t$ is a function which maps a $d$-dimensional random variable $X$ into a subset of the power set of the space $M_t$ of eligible portfolios at time $t$. The set $R_t(X)$ contains those portfolio vectors that can be used to cover the risk of the input portfolio.  Consider $\LdpF{T}$ with $p \in \lrsquare{0,+\infty}$ as the pre-image space and recall that $M_t\subseteq\LdpF{t}$.  In this paper we study risk measures
\[
R_t: \LdpF{T} \to \mathcal{P}\lrparen{M_t;(M_t)_+} := \lrcurly{D \subseteq M_t: D = D + (M_t)_+}.
\]
Conceptually the image space $\mathcal{P}\lrparen{M_t;(M_t)_+}$ contains the proper sets to consider.  This is since if some portfolio $m_t \in M_t$ covers the risk of a portfolio then so should any eligible portfolio that is (almost surely) greater than $m_t$, i.e. an element of $m_t + (M_t)_+$. The choice of the image space is further discussed in remark~\ref{rem image space}.

\begin{definition}
\label{defn_conditional}
A function $R_t: \LdpF{T} \to \mathcal{P}\lrparen{M_t;(M_t)_+}$ is a \textbf{\emph{conditional risk measure}} at time $t$ if it is
\begin{enumerate}
\item $M_t$-translative: $\forall m_t \in M_t: R_t\lrparen{X + m_t } = R_t(X) - m_t$;
\item $\LdpF{T}_{+}$-monotone: $Y-X \in \LdpF{T}_{+} \Rightarrow R_t(Y) \supseteq R_t(X)$;
\item finite at zero: $\emptyset \neq R_t(0) \neq M_t$.
\end{enumerate}
\end{definition}

Conceptually, $R_t(X)$ is the set of eligible portfolios which compensate for the risk of the portfolio $X$ at time $t$.  In this sense, $M_t$-translativity implies that if part of the risk is covered, only the remaining risk needs to be considered. This property in the scalar framework is sometimes called cash-invariance.  $\LdpF{T}_{+}$-monotonicity also has a clear interpretation: if a random vector $Y\in\LdpF{T}$ dominates another random vector $X\in\LdpF{T}$, then there should be more possibilities to compensate the risk of $Y$ (in particular cheaper ones) than for $X$.  Finiteness at zero means that there is an eligible portfolio at time $t$ which covers the risk of the zero payoff, but not all portfolios compensate for it.

Beyond these bare minimum of properties given in definition~\ref{defn_conditional}, there are additional properties that are desirable and are described in definition~\ref{defn_properties_conditional}.  Notable among these is $K_t$-compatibility (defined below).  This property allows for a portfolio manager to take advantage of trading opportunities
when assessing the risk of a position.  For such a concept we need a market model. Let us consider a market with proportional transaction costs as in \cite{K99,S04,KS09}, which is modeled by a sequence of solvency cones $\seq{K}$ describing a conical market model. $K_t$ is a solvency cone at time $t$ if it is an $\Ft{t}$-measurable cone such that for every $\omega \in \Omega$, $K_t(\omega)$ is a closed convex cone with $\mathbb{R}_+^d \subseteq K_t(\omega) \subsetneq \mathbb{R}^d$. $K_t$ can be generated  by the time $t$ bid-ask exchange rates between any two assets, for details see \cite{K99,S04,KS09} and examples~\ref{ex_solvency_no transaction} and~\ref{ex_solvency_transaction} below. The financial interpretation of the solvency cone at time $t$ is the set of positions which can be exchanged into a nonnegative portfolio at time $t$ by trading according to the prevailing bid-ask exchange rates. Let us denote by $K_t^{M_t}:=M_t\cap \LdpK{p}{t}{K_t}\subseteq \LdpF{t}$ the cone of solvent eligible portfolios only.
\begin{example}
\label{ex_solvency_no transaction}
In frictionless markets, the solvency cone $K_t(\omega)$ at time $t$ and in state $\omega \in \Omega$ is the positive halfspace defined by \[K_t(\omega) = \lrcurly{k \in \R^d:  \trans{S_t(\omega)}k\geq 0}\] where $S_t$ is the (random) vector of prices in some num\'{e}raire.
\end{example}
\begin{example}
\label{ex_solvency_transaction}
Assume that we have a sequence of (random) bid-ask matrices $\seq{\Pi}$, as in~\cite{S04}, where $\pi_t^{ij}$ is the (random) number of units of asset $i$ needed to purchase one unit of asset $j$ at time $t$ (where $i,j \in \{1,2,...,d\}$). Then the corresponding market model is given by the sequence of solvency cones $\seq{K}$ defined such that $K_t(\omega) \subseteq \R^d$ is the convex cone spanned by the unit vectors $e^i$ for every $i \in \{1,2,...,d\}$ and the vectors $\pi_t^{ij}(\omega) e^i - e^j$ for every $i,j \in \{1,2,...,d\}$.
\end{example}

\begin{definition}
\label{defn_properties_conditional}
A conditional risk measure $R_t: \LdpF{T} \to \mathcal{P}\lrparen{M_t;(M_t)_+}$ is called
\begin{itemize}
\item \textbf{\emph{convex}} if for all $X,Y \in \LdpF{T}$, for all $\lambda \in [0,1]$
\[R_t(\lambda X + (1-\lambda)Y) \supseteq \lambda R_t(X) + (1-\lambda)R_t(Y);\]
\item \textbf{\emph{conditionally convex}} if for all $X,Y \in \LdpF{T}$, for all $\lambda \in \LiF{t} \st 0 \leq\lambda\leq 1$
\[R_t(\lambda X + (1-\lambda)Y) \supseteq \lambda R_t(X) + (1-\lambda)R_t(Y);\]
\item \textbf{\emph{positive homogeneous}} if for all $X \in \LdpF{T}$, for all $\lambda \in \R_{++}$
\[R_t(\lambda X) = \lambda R_t(X);\]
\item \textbf{\emph{conditionally positive homogeneous}} if for all $X \in \LdpF{T}$, for all $\lambda \in \LiF{t}_{++}$
\[R_t(\lambda X) = \lambda R_t(X);\]
\item \textbf{\emph{subadditive}} if for all $X,Y \in \LdpF{T}$
\[R_t(X + Y) \supseteq R_t(X) + R_t(Y);\]
\item \textbf{\emph{coherent}} if it is convex and positive homogeneous, or subadditive and positive homogeneous;
\item \textbf{\emph{conditionally coherent}} if it is conditionally convex and conditionally positive homogeneous, or subadditive and conditionally positive homogeneous;
\item \textbf{\emph{normalized}} if $R_t(X) = R_t(X) + R_t(0)$
for every $X \in \LdpF{T}$;
\item \textbf{\emph{closed}} if the graph of $R_t$, defined by
\begin{equation}
\label{graph}
\operatorname{graph}R_t = \lrcurly{(X,u) \in \LdpF{T} \times M_t: \; u \in R_t(X)},
\end{equation}
is closed ($\sigma\lrparen{\LdiF{T},\LdoF{T}}$-closed if $p = \infty$);
\item \textbf{\emph{$K_t$-compatible}} (i.e. allowing trading at time $t$) if $R_t(X)=R_t(X)+K_t^{M_t}$ where $K_t$ is a solvency cone at time $t$;
\item \textbf{\emph{local}} if $1_A R_t\lrparen{1_A X} = 1_A R_t(X) $ for all $X \in \LdpF{T}$ and $A \in \Ft{t}$, where the stochastic indicator function is denoted by $1_A(\omega)$ being $1$ if $\omega\in A$ and $0$ if $\omega\notin A$.
\end{itemize}
\end{definition}

Before discussing the properties mentioned above, we will define dynamic risk measures.  A dynamic risk measure is a series of conditional risk measures.  In this way we can construct the process of risk compensating portfolios.
\begin{definition}
\label{defn_dynamic}
$\seq{R}$ is a \textbf{\emph{dynamic risk measure}} if $R_t$ is a conditional risk measure for every $t \in \{0,1,...,T\}$.
\end{definition}

\begin{definition}
\label{defn_properties_dynRM}
A dynamic risk measure $\seq{R}$ is said to have one of the properties given in definition~\ref{defn_properties_conditional} if $R_t$ has this property for every $t \in \{0,1,...,T\}$.

A dynamic risk measure $\seq{R}$ is called \textbf{\emph{market-compatible}} (i.e. trading is allowed at any time point $t \in \{0,1,...,T\}$) if $R_t$ is $K_t$-compatible for every $t \in \{0,1,...,T\}$.
\end{definition}

We now discuss the properties from definition~\ref{defn_properties_conditional} in the order given.

(Conditional) convexity is regarded as a useful property for dynamic risk measures because it defines a regulatory scheme which promotes diversification as discussed in \cite{FS02} in the scalar static case and \cite{HHR10} in the set-valued static case.  The values of a (conditionally) convex set-valued conditional risk measure $R_t$ are (conditionally) convex. The distinction between convexity and conditional convexity is whether the combination of portfolios can depend on the state of the market.  Trivially it can be seen that conditional convexity is a stronger property than convexity.  Typically, in the scalar framework (see e.g.~\cite{DS05,FP06,CDK06}), only conditional convexity is considered; though~\cite{KS05,BN04,BN08,BN09} consider convex local risk measures instead.

If $R_t$ is positive homogeneous then $R_t(0)$ is a cone and if $R_t$ is sublinear (positive homogeneous and subadditive), then $R_t(0)$ is a convex cone which is included in the recession cone of $R_t(X)$
for all $X \in \LdpF{T}$.
Therefore, as in the scalar framework, any closed coherent risk measure is normalized.

Let us give an interpretation of the normalization property. A normalized closed conditional risk measure $R_t$ satisfies $\lrparen{M_t}_+ \subseteq R_t(0)$ and $R_t(0) \cap \lrparen{M_t}_{--} = \emptyset$. Thus, nonnegative portfolios cover the risk of the zero payoff, but strictly negative portfolios cannot cover that risk. This clearly is the set-valued analog of the scalar normalization property given by $\rho(0)=0$.
 As it is typical in the set-valued framework, there is not a unique generalization of the normalization property for scalar risk measures to the set-valued case. Note that for set-valued risk measures normalization could be defined also in a `weaker' sense, as it was done in \cite{HH10,HR08}, by directly imposing $\lrparen{M_t}_+ \subseteq R_t(0)$ and $R_t(0) \cap \lrparen{M_t}_{--} = \emptyset$. In this paper, we will use the definition given in definition~\ref{defn_properties_conditional} for several reasons. First, it turns out to be the appropriate property when discussing time consistency of risk measures. Also, in the closed coherent set-valued case both notions coincide, see property~3.1 in \cite{JMT04}.  Second, any conditional risk measure (that is finite at zero) can be normalized in the following way.
The normalized version $\tilde{R}_t$ of a conditional risk measure $R_t$ is defined by
\[
\tilde{R}_t(X) = R_t(X) -^. R_t(0) := \lrcurly{u \in M_t: u +R_t(0) \subseteq R_t(X)}
\]
for every $X \in \LdpF{T}$.  The operation $A -^. B$ for sets $A,B$ is sometimes called the Minkowski difference (\cite{H50}) or geometric difference (\cite{P81}). This difference notation trivially shows how the normalization procedure introduced above relates to the normalized version for scalar risk measures given by $\rho_t(X) - \rho_t(0)$.  The normalized version $\tilde{R}_t(X)$ will also satisfy $(M_t)_+ \subseteq \tilde{R}_t(0)$ and $\tilde{R}_t(0) \cap \lrparen{M_t}_{--} = \emptyset$.

For convex duality results a closedness property is necessary.
Any set-valued function $R_t$ with a closed graph has closed values, i.e. $R_t(X)$ is a closed set for every $X \in \LdpF{T}$.

As previously mentioned, $K_t$-compatibility means that trading at time $t$ is allowed. If $u \in R_t(X)$ is a risk compensating portfolio for $X$ at time $t$, then also every portfolio in $u+K_t^{M_t}$ compensates for the risk of $X$.  This is accomplished by trading the new risk covering portfolio $v \in u + K_t^{M_t}$ into an element in $u+\LdpF{t}_+$, and thus into $u$ as well. Market-compatibility is the corresponding property for dynamic risk measures, when trading at each time point $t \in \{0,1,...,T\}$ is allowed.

The local property is a desirable property for dynamic set-valued risk measures. It has been studied for scalar dynamic risk measures for example in \cite{DS05,CK10,FKV11}. The local property ensures that the output of a risk measure at time $t$ at a specific state in $\Ft{t}$ only depends on the payoff in scenarios that can still be reached from this state, as we would expect of a risk compensating portfolio.

\begin{proposition}
Any conditionally convex risk measure $R_t: \LdpF{T} \to \mathcal{P}\lrparen{M_t;(M_t)_+}$ is local.
\end{proposition}
The proof is an adaption from the proof of lemma~3.4 in \cite{FKV11} with the convention that $0 \cdot R_t(X) = (M_t)_+$ for any $X \in \LdpF{T}$.
\begin{proof}
Let $X\in\LdpF{T}$ and $A\in\Ft{t}$. By convexity it is obvious that
\begin{align*}
R_t(1_A X) &\supseteq 1_A R_t(X) + 1_{A^c} R_t(0)
\\
&= 1_A R_t(1_A(1_AX) + 1_{A^c}X) + 1_{A^c} R_t(0)\supseteq 1_A R_t(1_A X) + 1_{A^c} R_t(0).
\end{align*}
  Multiplying through by $1_A$, then the left and right sides are equal, therefore $1_A R_t(1_A X) = 1_A R_t(X)$.
\end{proof}

\begin{remark}
\label{rem image space}
From $M_t$-translativity and monotonicity it follows that $R_t(X) = R_t(X) + \lrparen{M_t}_+$, which mathematically justifies the choice of the image space of a conditional risk measure to be $\mathcal{P}\lrparen{M_t;(M_t)_+}$ instead of the full powerset $2^{M_t}$.
The image space of a closed convex conditional risk measure is
\begin{equation*}
\mathcal{G}(M_t;(M_t)_+) = \lrcurly{D \subseteq M_t: D = \operatorname{cl}\operatorname{co}\lrparen{D + \lrparen{M_t}_+}}.
\end{equation*}
If $R_t$ is additionally $K_t$-compatible, then $R_t$ maps into
\begin{equation*}
\mathcal{G}(M_t;K_t^{M_t}) = \lrcurly{D \subseteq M_t: D = \operatorname{cl}\operatorname{co}\lrparen{D + K_t^{M_t}}}.
\end{equation*}
\end{remark}

\subsection{Dynamic acceptance sets}
Acceptance sets are intrinsically linked to risk measures.
A set-valued risk measure provides the portfolios which compensate for the risk of a contingent claim, whereas a portfolio is an element of the acceptance set if its risk does not need to be covered.  Typically, a conditional acceptance set at time $t$ is given by a regulator or risk manager.  We will show below in remark~\ref{rem_Rt_At} that there exists a bijective relation between risk measures and acceptance sets.
Expanding upon the definition given in \cite{HHR10} for the acceptance set of a static set-valued risk measures, we now define acceptance sets for dynamic set-valued risk measures.
\begin{definition}
\label{defn_acceptance}
$A_t \subseteq \LdpF{T}$ is a \textbf{\emph{conditional acceptance set}} at time $t$ if it satisfies the following:
\begin{enumerate}
\item \label{A1a} $M_t  \cap A_t \neq \emptyset $,
\item \label{A1b} $M_t  \cap \lrparen{\LdpF{T}\backslash A_t} \neq \emptyset $, and
\item $A_t + \LdpF{T}_{+} \subseteq A_t$.
\end{enumerate}
\end{definition}
The properties given in definition \ref{defn_acceptance} are the minimal requirements for a conditional acceptance set.  As noted in \cite{HHR10} these are conditions that any rational regulator would agree upon.  They imply that there is an eligible portfolio at time $t$ that is accepted by the regulator, but there exists an eligible portfolio which would itself require compensation to make it acceptable.  The last condition means that if a portfolio is acceptable to the regulator then any portfolio with at least as much of each asset will also be acceptable.
\begin{remark}
\label{rem_Rt_At}
It can be seen that there is a one-to-one relationship between acceptance sets and risk measures in the dynamic setting.  In particular, if we define the set
\[
    A_{R_t} := \lrcurly{X \in \LdpF{T}: 0 \in R_t(X) }
\]
for a given conditional risk measure $R_t$, then $A_{R_t}$ satisfies the properties of definition~\ref{defn_acceptance} and thus is a conditional acceptance set.
 On the other hand, defining
\begin{equation}
\label{RtAt}
    R_t^{A_t}(X) := \lrcurly{u \in M_t: X + u  \in A_t }
\end{equation}
for a given conditional acceptance set $A_t$ ensures that $R_t^{A_t}$ is a conditional risk measure.  Further, it can trivially be shown that $A_t = A_{R_t^{A_t}} $ and $R_t(\cdot) = R_t^{A_{R_t}}(\cdot)$. Equation \eqref{RtAt} is often referred to as the primal representation for the risk measure and provides the capital requirement interpretation of a risk measure.
\end{remark}

\begin{example}
\label{ex_worst_case}
The worst case risk measure is given by $R_t^{WC}(X) = \lrcurly{u \in M_t: X + u \in \LdpF{T}_+}$ with the acceptance set $\LdpF{T}_+$.
\end{example}

The following proposition is a list of corresponding properties between classes of conditional risk measures and classes of conditional acceptance sets.  If a risk measure $R_t$ has the (risk measure) property then $A_{R_t}$ has the corresponding (acceptance set) property, and vice versa: if an acceptance set $A_t$ has the (acceptance set) property then $R_t^{A_t}$ has the corresponding (risk measure) property.

\begin{proposition}
\label{propo A_t-R_t}
The following properties are in a one-to-one relationship for a conditional risk measure $R_t: \LdpF{T} \to \mathcal{P}\lrparen{M_t;(M_t)_+}$ and an acceptance set $A_t \subseteq \LdpF{T}$ at time $t$:
\begin{enumerate}
\item $R_t$ is $B$-monotone, and $A_t + B \subseteq A_t $, where $B \subseteq \LdpF{T}$;
\item \label{At-Rt_recession} $R_t$ maps into the set
\begin{equation*}
\mathcal{P}\lrparen{M_t;C} = \lrcurly{D \subseteq M_t: D = D + C },
\end{equation*}
and $A_t + C \subseteq A_t $, where $C \subseteq M_t$ with $0 \in \cl C$;
\item $R_t$ is (conditionally) convex, and $A_t$ is (conditionally) convex;
\item $R_t$ is (conditionally) positive homogeneous, and $A_t$ is a (conditional) cone;
\item $R_t$ is subadditive, and $A_t + A_t \subseteq A_t$;
\item $R_t$ is sublinear, and $A_t$ is a convex cone;
\item $R_t$ has a closed graph, and $A_t$ is closed;
\item $R_t(X) \neq \emptyset $ for all $X \in \LdpF{T}$, and $\LdpF{T} = A_t + M_t $;
\item $R_t(X) \neq M_t $ for all $X \in \LdpF{T}$, and $\LdpF{T} = \lrparen{\LdpF{T} \backslash A_t} + M_t $.
\end{enumerate}
\end{proposition}
\begin{proof}
All these properties follow from adaptations from proposition 6.5 in \cite{HHR10}.
\end{proof}
It can be seen that property~\ref{At-Rt_recession} in proposition~\ref{propo A_t-R_t} corresponds to normalization ($C = R_t(0) = A_t \cap M_t$) and $K_t$-compatibility ($C = K_t^{M_t}$).

\begin{example}[{[Example~\ref{ex_worst_case} continued]} ]
The worst case risk measure is a closed coherent risk measure.  In fact, the worst case acceptance set is the smallest closed normalized acceptance set.
\end{example}

\section{Time consistency}
\label{sec_time_consistent}
In this section we will study two different time consistency properties for set-valued dynamic risk measures.  Most generally, time consistency properties concern how risk relates through time.  A risk manager would want the risk compensating portfolios to not conflict with each other as time progresses.  In the scalar framework time consistency is defined by
\[
	\rho_{t+1}(X) \geq \rho_{t+1}(Y)  \Rightarrow \rho_t(X) \geq \rho_t(Y)
\]
for all times $t$ and portfolios $X,Y \in \LpF{T}$.  This means that if it is known, a priori, that one portfolio is less risky than another in (almost) every state of the world, this relation holds backwards in time.  In particular, by compensating for the risks of $X$ then the risks of $Y$ could also be covered.

A key result in the scalar framework is the equivalence between time consistency and a recursive representation given by $\rho_t(X) = \rho_t(-\rho_{t+1}(X))$ for all portfolios $X$ and times $t$.  This recursive form relates to Bellman's principle which provides an efficient method for calculation.  Time consistency, the recursive formulation, and other equivalent properties has been studied in great detail in the scalar case in papers such as \cite{AD07,R04,DS05,CDK06,RS05,BN08,FP06,CS09,CK10}.

When generalizing to the set-valued framework we present two different properties which are analogous to time consistency in the scalar framework.  The first property we present, which we call `time consistency' is the intuitive generalization of scalar time consistency.  The other property, which we call `multi-portfolio time consistency,' is shown to be equivalent to the recursive form for set-valued risk measures.  We will demonstrate how both notions are related to each other and under which conditions they coincide.

\begin{definition}
\label{defn_tc}
A dynamic risk measure $\seq{R}$ is called \textbf{\emph{time consistent}} if for all times $t \in \{0,1,...,T-1\}$ and $X, Y \in \LdpF{T}$ it holds
\begin{equation*}
R_{t+1}(X) \subseteq R_{t+1}(Y) \Rightarrow R_t(X) \subseteq R_t(Y).
\end{equation*}
\end{definition}

As in the scalar case described previously, time consistency implies that if it is known at some future time that every eligible portfolio which covers the risk of a claim $X$ also would cover the claim $Y$, then any prior time if a portfolio covers the risk of $X$ it would also do so for $Y$.

The recursive form for set-valued risk measures is understood pointwise.  That is, \[R_t(-R_{t+1}(X)) := \bigcup_{Z \in R_{t+1}(X)} R_t(-Z).\]  Then the generalization of the recursion from the scalar case is given by $R_t(X) = R_t(-R_{t+1}(X))$.  This can be seen as a set-valued Bellman's principle.

\subsection{Multi-portfolio time consistency}
We introduce multi-portfolio time consistency in definition~\ref{defn_mptc} below.  In theorem~\ref{thm_equiv_tc} we show that multi-portfolio time consistency is equivalent to the recursive form for normalized risk measures.  In example~\ref{ex_tc}, we show that time consistency is a weaker property than multi-portfolio time consistency. Furthermore, we will show in lemma~\ref{lemma_tc2recursive} and remark~\ref{rem tc=mptc scalar} below that both concepts coincide in the scalar case.

\begin{definition}
\label{defn_mptc}
A dynamic risk measure $\seq{R}$ is called \textbf{\emph{multi-portfolio time consistent}} if for all times $t \in \{0,1,...,T-1\}$ and all sets $A,B \subseteq \LdpF{T}$ the implication
\begin{equation}
\bigcup_{X \in A} R_{t+1}(X) \subseteq \bigcup_{Y \in B} R_{t+1}(Y) \Rightarrow \bigcup_{X \in A} R_t(X) \subseteq \bigcup_{Y \in B} R_t(Y)
\end{equation}
is satisfied.
\end{definition}

Conceptually multi-portfolio time consistency of a dynamic risk measure means that if a market sector, i.e. a collection of portfolios, is more risky than another sector at some future time, then the same must be true for prior times.  A market sector $A$ is considered more risky than sector $B$ if any portfolio which compensates for the risk of a claim in $A$ would also cover the risk of a claim in $B$.  Trivially a risk measure which is multi-portfolio time consistent is also time consistent, but the converse is not true in general (see example~\ref{ex_tc}).

We prove now that the recursive structure is equivalent to multi-portfolio time consistency for normalized risk measures.  Additionally, given extra assumptions on the sets of eligible assets, we demonstrate the equivalence between these properties and a property on the acceptance sets.  We will later show, in section~\ref{subsec composition}, how to use these properties to construct multi-portfolio time consistent risk measures.

For the remainder of this section the convention that $A_t = A_{R_t}$ for a conditional risk measure $R_t$ at time $t$ will be used.  The $s$-stepped acceptance set at time $t$ is defined as below.
\begin{definition}
The set $A_{t,t+s} \subseteq \LdpF{t+s}$ defined by
\begin{equation*}
A_{t,t+s} = \lrcurly{X \in \LdpF{t+s}: 0 \in R_t(X )}=A_t\cap\LdpF{t+s}
\end{equation*}
is called \textbf{\emph{$s$-stepped acceptance set}} at time $t$.
\end{definition}
Furthermore, denote the intersection of the $s$-stepped acceptance set at time $t$ and the set of time $t+s$ eligible portfolios by $A_{t,t+s}^{M_{t+s}}$, that is,
\[
    A_{t,t+s}^{M_{t+s}} = A_{t,t+s} \cap M_{t+s}=A_t\cap M_{t+s}.
\]

\begin{theorem}
\label{thm_equiv_tc}
For a normalized dynamic risk measure $\seq{R}$ the following are equivalent:
\begin{enumerate}
\item \label{thm_equiv_tctc}$\seq{R}$ is multi-portfolio time consistent;
\item \label{thm_equiv_eq} for every time $t \in \{0,1,...,T-1\}$ and $A,B \subseteq \LdpF{T}$
\begin{equation*}
\bigcup_{X \in A} R_{t+1}(X) = \bigcup_{Y \in B} R_{t+1}(Y) \Rightarrow \bigcup_{X \in A} R_t(X) = \bigcup_{Y \in B} R_t(Y);
\end{equation*}
\item \label{thm_equiv_recursive} $R_t$ is recursive, that is for every time $t \in \{0,1,...,T-1\}$
\begin{equation}
\label{recursive}
R_t(X) = \bigcup_{Z \in R_{t+1}(X)} R_t(-Z)=:R_t(-R_{t+1}(X)).
\end{equation}
\end{enumerate}
If additionally $M_t \subseteq M_{t+1}$ for every time $t \in \{0,1,...,T-1\}$ then
all of the above is also equivalent to
\begin{enumerate}
\item[(iv)] for every time $t \in \{0,1,...,T-1\}$
\[A_t = A_{t+1} + A_{t,t+1}^{M_{t+1}}.\]
\end{enumerate}
\end{theorem}

\begin{proof}
It can trivially be seen that property \ref{thm_equiv_tctc} implies property \ref{thm_equiv_eq}.  By $\seq{R}$ normalized it follows that for every $X \in \LdpF{T}$ and $t \in \{0,1,...,T-1\}$ it holds
\begin{align*}
\bigcup_{Z \in R_{t+1}(X)} R_{t+1}(-Z) & = \bigcup_{Z \in R_{t+1}(X)} (R_{t+1}(0) + Z)\\
& = R_{t+1}(0) + R_{t+1}(X) = R_{t+1}(X).
\end{align*}
Thus by property \ref{thm_equiv_eq} and setting $A = \{X\}$ and $B = -R_{t+1}(X) $, the recursive form defined in equation~\eqref{recursive} is derived and thus property \ref{thm_equiv_recursive}.  It remains to show that \ref{thm_equiv_recursive} implies multi-portfolio time consistency as defined in definition~\ref{defn_mptc}, i.e. property~\ref{thm_equiv_tctc}.  If $A,B \subseteq \LdpF{T}$ such that $\bigcup_{X \in A} R_{t+1}(X) \subseteq \bigcup_{Y \in B} R_{t+1}(Y)$ and let $\seq{R}$ be recursive (as defined in equation~\eqref{recursive}) then
\begin{align*}
\bigcup_{X \in A} R_t(X) & = \bigcup_{X \in A} \bigcup_{Z \in R_{t+1}(X)} R_t(-Z) = \bigcup_{Z \in \cup_{X \in A} R_{t+1}(X)} R_t(-Z ) \\
& \subseteq \bigcup_{Z \in \cup_{Y \in B} R_{t+1}(Y)} R_t(-Z)  = \bigcup_{Y \in B} R_t(Y).
\end{align*}

Finally if $M_t \subseteq M_{t+1}$ for every time $t \in \{0,1,...,T-1\}$ then by lemma \ref{lemma_stepped} below property~(iv) is equivalent to the recursive form, i.e. is equivalent to $\seq{R}$ being multi-portfolio time consistent.
\end{proof}

As mentioned at the beginning of section~\ref{sec_dynamic_defn}, the extra condition needed for property~(iv) in theorem~\ref{thm_equiv_tc}, that is $M_t \subseteq M_{t+s}$ for $s,t \geq 0$ with $t,t+s \in \{0,1,...,T\}$, is satisfied in many situations that are usually considered as this is met if $M_t = M_{0}$ $\P$-almost surely.  In particular, in the scalar framework this condition is always satisfied, see example~\ref{ex_eligible}.

\begin{remark}
As it is trivially noticed using the acceptance set definition (and also is implicitly understood in the other definitions) for multi-portfolio time consistency, the choice of eligible assets would impact whether a risk measure is multi-portfolio time consistent.  Therefore, it is possible that under one choice of eligible portfolios the risk measure is recursive, but under another choice of eligible portfolios that same risk measure is not recursive.
\end{remark}

\begin{example}[{[Example~\ref{ex_worst_case} continued]} ]
For the worst case risk measure, the acceptance sets are given by $A_t = \LdpF{T}_+$ for all times $t$, therefore $A_t = A_{t,t+1}^{M_{t+1}} + A_{t+1}$ for all times and any choice of eligible set $M_{t+1}$.  Thus, if $M_t \subseteq M_{t+1}$ for all times then the worst case risk measure is recursive and multi-portfolio time consistent.
\end{example}

The following lemma is used in the proof of theorem~\ref{thm_equiv_tc}. It is the set-valued generalization of
lemma~4.3 in \cite{FP06} or lemma~11.14 in \cite{FS11} and gives a relationship between the sum of conditional acceptance sets and properties of the corresponding risk measure.

\begin{lemma}
\label{lemma_stepped}
Let $\seq{R}$ be a dynamic risk measure.  Let $s,t \geq 0$ such that $t,t+s \in \{0,1,...,T\}$ and let $X \in \LdpF{T}$ and $D \subseteq \LdpF{T}$. It holds
\begin{enumerate}
\item \label{stepped_acceptance_1} $X \in A_{t+s} + D \cap M_{t+s}  \Leftrightarrow -R_{t+s}(X) \cap D \neq \emptyset$;
\item \label{stepped_acceptance_2} \begin{enumerate}
\item $R_t(X) \subseteq \bigcup_{Z \in R_{t+s}(X)} R_t(-Z ) \Rightarrow A_t \subseteq A_{t+s} + A_{t,t+s}^{M_{t+s}} $;
\item If additionally $M_t \subseteq M_{t+s}$ then $A_t \subseteq A_{t+s} + A_{t,t+s}^{M_{t+s}}  \Rightarrow R_t(X) \subseteq \bigcup_{Z \in R_{t+s}(X)} R_t(-Z )$;
\end{enumerate}
\item \label{stepped_acceptance_3} \begin{enumerate}
\item $R_t(X) \supseteq \bigcup_{Z \in R_{t+s}(X)} R_t(-Z ) \Rightarrow A_t \supseteq A_{t+s} + A_{t,t+s}^{M_{t+s}} $;
\item If additionally $M_t \subseteq M_{t+s}$ then $A_t \supseteq A_{t+s} + A_{t,t+s}^{M_{t+s}}  \Rightarrow R_t(X) \supseteq \bigcup_{Z \in R_{t+s}(X)} R_t(-Z )$.
\end{enumerate}
\end{enumerate}
\end{lemma}

\begin{proof}
\begin{enumerate}
\item \begin{enumerate}
\item[($\Rightarrow$)] Given that $X \in A_{t+s} + D \cap M_{t+s} $ then $X = X_{t+s} + X^{D} $ such that $X_{t+s} \in A_{t+s}$ and $X^{D} \in D \cap M_{t+s}$.  Therefore it can
    be seen that $R_{t+s}(X) = R_{t+s}\lrparen{X_{t+s} + X^{D} }$, and by $X^{D} \in M_{t+s}$ and the $M_{t+s}$-translative property, $-X^{D} \in R_{t+s}\lrparen{X_{t+s}} - X^{D} = R_{t+s}(X)$.  Then $X^{D} \in -R_{t+s}(X)$ and by assumption $X^{D} \in D \cap M_{t+s}$.
\item[($\Leftarrow$)] Given that there exists a $Y \in -R_{t+s}(X)$ such that $Y \in D \cap M_{t+s}$ and trivially $X = X - Y + Y$, then $R_{t+s}(X - Y ) = R_{t+s}(X) + Y \ni 0$ since $-Y \in R_{t+s}(X)$.  Therefore $X-Y  \in A_{t+s}$ and $X \in A_{t+s} + D \cap M_{t+s} $.
\end{enumerate}

\item \begin{enumerate}
\item Let $X \in A_t$ then $0 \in R_t(X) \subseteq \bigcup_{Z \in R_{t+s}(X)} R_t(-Z)$.  This implies that there exists as $Y \in -R_{t+s}(X)$ such that $0 \in R_t(Y )$, and thus $Y \in A_{t,t+s}^{M_{t+s}}$.  Therefore $X \in A_{t+s} + A_{t,t+s}^{M_{t+s}} $ by condition \ref{stepped_acceptance_1}.
\item Let $X \in \LdpF{T}$ and $Y \in R_t(X)$ then $X+Y  \in A_t \subseteq A_{t+s} + A_{t,t+s}^{M_{t+s}} $.  By condition \ref{stepped_acceptance_1}, there exists a $\hat{Z} \in -R_{t+s}(X+Y)$ such that $\hat{Z} \in A_{t,t+s}^{M_{t+s}}$.  But then there exists a $Z \in R_{t+s}(X)$ such that $\hat{Z} = -\lrparen{Z - Y}$ using translative property of $R_{t+s}$ and because $M_t \subseteq M_{t+s}$.  Since $\hat{Z} \in A_{t,t+s}^{M_{t+s}}$ it holds $0 \in R_t(\hat{Z}) = R_t\lrparen{\lrparen{-Z+Y}} = R_t\lrparen{-Z}-Y$.  Therefore $Y \in R_t\lrparen{-Z}$.  Thus for all $Y \in R_t(X)$ there exists a $Z \in R_{t+s}(X)$ such that $Y \in R_t(-Z)$.  This can be rewritten as $R_t(X) \subseteq \bigcup_{Z \in R_{t+s}(X)} R_t(-Z )$.
\end{enumerate}

\item \begin{enumerate}
\item Let $X \in A_{t+s} + A_{t,t+s}^{M_{t+s}} $, then there exists a $Y \in -R_{t+s}(X)$ such that $Y \in A_{t,t+s}^{M_{t+s}}$ by condition \ref{stepped_acceptance_1}.  Therefore, $R_t(X) \supseteq \bigcup_{Z \in R_{t+s}(X)} R_t(-Z ) \supseteq R_t(Y ) \ni 0$ and thus $X \in A_t$.
\item Let $X \in \LdpF{T}$, any $Z \in R_{t+s}(X)$, and any $Y \in R_t(-Z )$, then $Y +X = Y  -Z +Z +X$.  In particular, $X+Z \in A_{t+s}$ since $Z \in R_{t+s}(X)$ if and only if $0 \in R_{t+s}(X)-Z = R_{t+s}(X+Z )$, and $Y-Z \in A_{t,t+s}^{M_{t+s}}$ since $Y \in R_t(-Z)$ if and only if $0 \in R_t(-Z)-Y = R_t(Y-Z)$ and $Y-Z \in M_{t+s}$ by $Y \in M_t$ and $-Z \in M_{t+s}$ with $M_t + M_{t+s} \subseteq M_{t+s}$ (by $M_t \subseteq M_{t+s}$).  Therefore, $Y+X \in A_{t+s} + A_{t,t+s}^{M_{t+s}} \subseteq A_t$.  This implies $0 \in R_t(Y+X) = R_t(X)-Y$.  Therefore $Y \in R_t(X)$ and for every $Z \in R_{t+s}(X)$ it follows that $R_t(-Z) \subseteq R_t(X)$.  Therefore $\bigcup_{Z \in R_{t+s}(X)} R_t(-Z) \subseteq R_t(X)$.
\end{enumerate}
\end{enumerate}
\end{proof}

\begin{remark}
\label{rem_recursiv_mptc}
As can be seen from lemma~\ref{lemma_stepped}, the normalization property is not used for the equivalences between properties~(iii) and (iv) in theorem~\ref{thm_equiv_tc} under assumption $M_t \subseteq M_{t+1}$ for all $t$. Furthermore, if a dynamic risk measure $\seq{R}$ that is not normalized follows the recursive form defined in equation~\eqref{recursive}, then $\seq{R}$ is multi-portfolio time consistent (but not necessarily vice versa).
\end{remark}

Now that we have given the main result on multi-portfolio time consistency in this paper, we are interested in
how multi-portfolio time consistency and time consistency relate to each other.
First, we show that the two properties are not equivalent in general (see example~\ref{ex_tc}).  Second, we give sufficient conditions for these properties to coincide in lemma~\ref{lemma_tc2recursive}.  Remark~\ref{rem tc=mptc scalar} demonstrates that these sufficient conditions are always satisfied in the scalar framework.

\begin{example}
\label{ex_tc}
 Consider a one-period model with $t \in \lrcurly{0,T}$. Let $A_0 = \lrcurly{X \in \LdpF{T}: X_1 \in \LpF{T}_+}$ and $A_{T} = \lrcurly{X \in \LdpF{T}: X_1 \in \LpF{T}_{++}}$. Let $M_{t} = \lrcurly{X \in \LdpF{t}: \forall j \in \lrcurly{2,3,...,d}: X_j = 0}$ for $t \in \lrcurly{0,T}$. Clearly, $A_0$ and $A_T$ satisfy the properties of a normalized acceptance set (definition~\ref{defn_acceptance} and proposition~\ref{propo A_t-R_t}~(ii) for $C=A_t\cap M_t$, $t \in \lrcurly{0,T}$) and denote the corresponding dynamic risk measure $R_t(X) := \lrcurly{u \in M_t: X + u  \in A_t }$ for $t \in \lrcurly{0,T}$. Then, it holds $A_{0}\supsetneq A_{T} + A_{0,T}^{M_{T}} = A_{T}$ since $A_{0,T}^{M_{T}} = \lrcurly{X \in M_{T}: X_1 \geq 0}$. Thus, the risk measure $\seq{R}$ is not multi-portfolio time consistent by theorem~\ref{thm_equiv_tc}.
 However, $\seq{R}$ is time consistent since $R_{T}(X) \subseteq R_{T}(Y)$ if and only if $Y_1 - X_1 \in \LpF{T}_+$.  This implies $R_0(X) = R_0\lrparen{\lrsquare{X_1,0,...,0}} \subseteq R_0\lrparen{\lrsquare{Y_1,0,...,0}} = R_0(Y)$ by monotonicity and the definition of $A_0$.
\end{example}

As we have shown with this example, the recursive form and time consistency are not equivalent in general.  We now consider sufficient conditions for the two properties to be equivalent.

\begin{lemma}
\label{lemma_tc2recursive}
Let $\seq{R}$ be a normalized time consistent dynamic risk measure such that for all times $t \in \{0,1,...,T-1\}$ and every $X \in \LdpF{T}$ there exists a  $\hat{Z} \in R_{t+1}(X)$ such that
$R_{t+1}(-\hat{Z}) \supseteq R_{t+1}(X)$ (or, equivalently $R_{t+1}(X)=\hat{Z}+R_{t+1}(0)$).
Then, $\seq{R}$ is multi-portfolio time consistent.
\end{lemma}

\begin{proof}
Let $X \in \LdpF{T}$, $t \in \{0,1,...,T-1\}$ and $Z\in R_{t+1}(X)$ arbitrarily chosen. Then, $M_{t+1}$-translativity and normalization of $R_{t+1}$ imply $R_{t+1}(-Z)=R_{t+1}(0)+Z\subseteq R_{t+1}(X)$, hence $R_t(-Z) \subseteq R_t(X)$ by time consistency. Thus,
$ \bigcup_{Z \in R_{t+1}(X)} R_t(-Z)\subseteq R_t(X)$.
On the other hand, by assumption there exists a  $\hat{Z} \in R_{t+1}(X)$ such that
$R_{t+1}(-\hat{Z}) \supseteq R_{t+1}(X)$. Hence, $R_{t}(-\hat{Z}) \supseteq R_{t}(X)$ by time consistency and
$ \bigcup_{Z \in R_{t+1}(X)} R_t(-Z)\supseteq R_t(X)$.
\end{proof}

\begin{remark}
\label{rem tc=mptc scalar}
For a scalar normalized time consistent dynamic risk measure $\seq{\rho}$ with $\rho_t: \LpF{T} \to \LpF{t}$, the corresponding set-valued dynamic risk measure $\seq{R^{\rho}}$ defined on $\LpF{T}$ is given by $R^{\rho}_t(X)=\rho_t(X)+\LpF{t}_+$ and thus automatically satisfies the assumptions of lemma \ref{lemma_tc2recursive}.  This shows that in the scalar case time consistency is equivalent to multi-portfolio time consistency.
\end{remark}

Now we will take a more in depth look at market-compatibility for dynamic risk measures. For static risk measures as discussed in~\cite{HHR10, HRY12}, the right concept of market-compatibility is given by  $A_0 = A_0 + \sum_{t=0}^T \LdpK{p}{t}{K_t}$.  This might look different from the definition of market-compatibility for dynamic risk measures in the present paper, but it turns out that for multi-portfolio time consistent risk measures both notions coincides, which justified the use of the same name.
\begin{lemma}
\label{lemma_market_comp}
Let $\seq{R}$ be a normalized multi-portfolio time consistent dynamic risk measure with $M_t \subseteq M_{t+1}$ for every time $t \in \lrcurly{0,1,...,T-1}$. Then, $\seq{R}$ is market-compatible if and only if
$A_t = A_t + \sum_{\tau = t}^{T}K_{\tau}^{M_{\tau}}$ at each time $t$.
\end{lemma}
\begin{proof}
If $\seq{R}$ is market-compatible then by proposition~\ref{propo A_t-R_t} and since $0\in K_t^{M_t}$ it holds $A_t = A_t + K_t^{M_t}$ for every $t$.  By multi-portfolio time consistency $A_t = A_{t,t+1}^{M_{t+1}} + A_{t+1}$, therefore, $A_t =  A_{t,t+1}^{M_{t+1}} + A_{t+1} + K_t^{M_t}$.  Doing this argument for each time $t+1$ through $T-1$ and using $A_T = A_T + K_T^{M_T}$, it follows that $A_t = \sum_{\tau = t}^{T-1} A_{\tau,\tau+1}^{M_{\tau+1}} + A_{T} + \sum_{\tau = t}^{T} K_{\tau}^{M_{\tau}} = A_t + \sum_{\tau = t}^{T} K_{\tau}^{M_{\tau}}$ for any time $t$.

For the reverse implication let us assume $A_t = A_t + \sum_{\tau = t}^{T}K_{\tau}^{M_{\tau}}$ for each $t$. Fix some time $t < T$. It follows that
\begin{align*}
A_t & = A_t + \sum_{\tau = t}^{T}K_{\tau}^{M_{\tau}}
 =  A_{t,t+1}^{M_{t+1}} + A_{t+1} + K_{t}^{M_{t}} + \sum_{\tau = t+1}^{T}K_{\tau}^{M_{\tau}}\\
& = A_{t,t+1}^{M_{t+1}} + A_{t+1} + K_{t}^{M_{t}}
 =  A_t + K_{t}^{M_{t}}.
\end{align*}
And trivially at time $T$ this implies that $A_T = A_T + K_T^{M_T}$.  Therefore $\seq{R}$ is market-compatible.
\end{proof}

\begin{remark}
Lemma~\ref{lemma_market_comp} can be generalized in the following way. For $t \in \{0,1,...,T\}$ let $A_t,C_t,D_t \subseteq \LdpF{T}$ with $A_t = A_{t+1} + D_t$ for all $t \in \{0,1,...,T-1\}$. Then the following three statements are equivalent.
\begin{enumerate}
\item $A_t = A_t + C_t$ for all times $t \in \{0,1,...,T\}$,
\item $A_t = A_t + C_s$ for all times $t,s \in \{0,1,...,T\}$ with $s \geq t$,
\item $A_t = A_t + \sum_{\tau = t}^{T} C_{\tau}$ for all times $t \in \{0,1,...,T\}$.
\end{enumerate}
\end{remark}

\subsection{Composition of one-step risk measures}
\label{subsec composition}
Since multi-portfolio time consistency is a restrictive property for risk measures (in the scalar framework both value at risk and average value at risk are not time consistent, and in section~\ref{section_avar} we show that the set-valued average value at risk is not multi-portfolio time consistent either), we would like to have a way to construct multi-portfolio time consistent versions of any risk measure.  As in section~2.1 in \cite{CK10} and section~4 in \cite{CS09}, a (multi-portfolio) time consistent version of any scalar dynamic risk measure can be created through backwards recursion.
The same is true for set-valued risk measures as discussed in the following proposition.
\begin{proposition}
\label{lemma_gen_tc}
Let $\seq{R}$ be a dynamic risk measure on $\LdpF{T}$ and let $M_t \subseteq M_{t+1}$ for every time $t \in \lrcurly{0,1,...,T-1}$, then $(\tilde{R}_t)_{t=0}^T$ defined for all $X \in \LdpF{T}$ by
\begin{align}
\label{eqn_composed_final} \tilde{R}_{T}(X) & = R_{T}(X),\\
\label{eqn_composed} \forall t \in \lrcurly{0,1,...,T-1}: \; \tilde{R}_t(X) & = \bigcup_{Z \in \tilde{R}_{t+1}(X)} R_t(-Z)
\end{align}
is multi-portfolio time consistent. Furthermore, $(\tilde{R}_t)_{t=0}^T$ satisfies properties (i) and (ii) in definition~\ref{defn_conditional} of dynamic risk measures, but may fail to be finite at zero.  Additionally, if $\seq{R}$ is (conditionally) convex (coherent) then $(\tilde{R}_t)_{t=0}^T$ is (conditionally) convex (coherent).
\end{proposition}
\begin{proof}
Let $t \in \lrcurly{0,1,...,T-1}$ and $A,B \subseteq \LdpF{T}$ such that $\bigcup_{X \in A} \tilde{R}_{t+1}(X) \subseteq \bigcup_{Y \in B} \tilde{R}_{t+1}(Y)$, then
\begin{align*}
\bigcup_{X \in A} \tilde{R}_t(X) & = \bigcup_{X \in A} \bigcup_{Z \in \tilde{R}_{t+1}(X)} R_t(-Z)= \bigcup_{Z \in \bigcup_{X \in A} \tilde{R}_{t+1}(X)} R_t(-Z)\\
& \subseteq \bigcup_{Z \in \bigcup_{Y \in B} \tilde{R}_{t+1}(Y)} R_t(-Z) = \bigcup_{Y \in B} \tilde{R}_t(Y).
\end{align*}
Thus, $(\tilde{R}_t)_{t=0}^T$ is multi-portfolio time consistent.
The assumption $M_t \subseteq M_{t+1}$ ensures $\tilde{R}_t$ to be $M_t$-translative for all $t$. $\LdpF{T}_{+}$-monotonicity follows from the corresponding property for $R_t$.

If $\seq{R}$ is (conditionally) convex (positive homogeneous) then by backwards induction $(\tilde{R}_t)_{t=0}^T$ is (conditionally) convex (positive homogeneous) .
\end{proof}

$ (\tilde{R}_t)_{t=0}^T$ defined as in equations~\eqref{eqn_composed_final}, \eqref{eqn_composed} is multi-portfolio time consistent, but not necessarily  normalized or finite at zero.
If $ (\tilde{R}_t)_{t=0}^T$ is normalized, then $ (\tilde{R}_t)_{t=0}^T$ is recursive itself, see also remark~\ref{rem_recursiv_mptc}. Thus, we are interested to find conditions under which $ (\tilde{R}_t)_{t=0}^T$ is normalized, or  finite at zero.
\begin{proposition}
\label{cor_composed_normalized}
Let $\seq{R}$ and $\seq{M}$ be as in proposition~\ref{lemma_gen_tc} and let $ (\tilde{R}_t)_{t=0}^T$ be defined as in \eqref{eqn_composed_final}, \eqref{eqn_composed}. Then, $(\tilde{R}_t)_{t=0}^T$ is normalized and finite at zero if $\seq{R}$ is normalized and either of the following are true:
\begin{enumerate}
\item $R_t(0) = (M_t)_+$ for every time $t$, or
\item \label{cor_composed_tc} $\seq{R}$ is a closed and time consistent risk measure.
\end{enumerate}
If $\seq{R}$ is a normalized closed coherent risk measure then $(\tilde{R}_t)_{t=0}^T$ is normalized, coherent and for $t\in\{0,1,\dots,T\}$ either finite at zero or $\tilde{R}_t(X) \in \{\emptyset,M_t\}$ for every $X \in \LdpF{T}$.
\end{proposition}
\begin{proof}
If $R_T$ is normalized then it immediately follows that $\tilde{R}_T$ is normalized as well, indeed $\tilde{R}_T(0) = R_T(0)$.  Thus, using backwards induction, we want to show that $\tilde{R}_t(X) = \tilde{R}_t(X) + \tilde{R}_t(0)$.
\begin{enumerate}
\item Let $\tilde{R}_{t+1}(0) = R_{t+1}(0) = (M_{t+1})_+$.  Therefore $Z \in \tilde{R}_{t+1}(0)$ implies $R_t(-Z) \subseteq R_t(0)$ by $\LdpF{T}_{+}$-monotonicity.  Therefore \[\tilde{R}_t(0) = \bigcup_{Z \in \tilde{R}_{t+1}(0)} R_t(-Z) = R_t(0) = (M_t)_+.\]  Trivially it then holds that $\tilde{R}_t(X) = \tilde{R}_t(X) + \tilde{R}_t(0)$.

\item Let $\tilde{R}_{t+1}(0) = R_{t+1}(0)$.  By lemma~\ref{tc_normalized} below $R_t$ is $R_{t+1}(0)$-monotone.  Therefore it immediately follows that $R_t(-Z) \subseteq R_t(0)$ for every $Z \in \tilde{R}_{t+1}(0)$ (with $0 \in \tilde{R}_{t+1}(0)$ by $R_{t+1}$ closed and normalized), and \[\tilde{R}_t(0) = \bigcup_{Z \in \tilde{R}_{t+1}(0)} R_t(-Z) = R_t(0).\]  $\tilde{R}_t$ is normalized since $\tilde{R}_t(X) + \tilde{R}_t(0) = \bigcup_{Z \in \tilde{R}_{t+1}(X)} [R_t(-Z) + R_t(0)] = \bigcup_{Z \in \tilde{R}_{t+1}(X)} R_t(-Z) = \tilde{R}_t(X)$ by $R_t$ normalized.
\end{enumerate}

If $\seq{R}$ is a normalized closed coherent risk measure then $(\tilde{R}_t)_{t=0}^T$ is coherent by proposition~\ref{lemma_gen_tc}, and thus $\tilde{R}_t(X) \supseteq \tilde{R}_t(X) + \tilde{R}_t(0)$ by subadditivity.  To show the other direction assume that $0 \in \tilde{R}_{t+1}(0)$ (by $0 \in \tilde{R}_T(0)$ and backwards induction).  It follows that $\tilde{R}_t(0) = \bigcup_{Z \in \tilde{R}_{t+1}(0)} R_t(-Z) \supseteq R_t(0) \ni 0$ since $0 \in R_t(0)$ by closed and normalized. This implies $\tilde{R}_t(X) + \tilde{R}_t(0) \supseteq \tilde{R}_t(X)$.  Therefore $\tilde{R}_t$ is a normalized risk measure and $\tilde{R}_t(0) \neq \emptyset$.  However, it still may be the case that $\tilde{R}_t(0) = M_t$.  By normalization if $\tilde{R}_t(0) = M_t$ then $\tilde{R}_t(X) \in \{\emptyset,M_t\}$ for any $X \in \LdpF{T}$.
\end{proof}

Property~\ref{cor_composed_tc} in proposition~\ref{cor_composed_normalized} above shows that in the set-valued framework, even though time-consistency is not equivalent to multi-portfolio time consistency, it can be a useful property for the creation of multi-portfolio time consistent risk measures.

\begin{lemma}
\label{tc_normalized}
If $\seq{R}$ is a time consistent risk measure and $R_{\tau}$ is normalized for some time $\tau$, then $R_t$ is $R_{\tau}(0)$-monotone for any time $t \leq \tau$.
\end{lemma}
\begin{proof}
Let $X,Y \in \LdpF{T}$ such that $Y - X \in R_{\tau}(0)$, then $R_{\tau}(Y) = R_{\tau}(Y - X + X) = R_{\tau}(X) - (Y - X) = R_{\tau}(X) + R_{\tau}(0) - (Y - X) \supseteq R_{\tau}(X)$ by $Y-X \in M_{\tau}$, $R_{\tau}$ normalized, and $0 \in R_{\tau}(0) - (Y - X)$.  Then by time consistency, $R_{\tau}(Y) \supseteq R_{\tau}(X)$ implies $R_t(Y) \supseteq R_t(X)$ for any $t \leq \tau$, and therefore $R_t$ is $R_{\tau}(0)$-monotone.
\end{proof}

The following corollary provides a possibility to construct multi-portfolio time consistent risk measures by backward composition using the (one step) acceptance sets of any dynamic risk measure.

\begin{corollary}
\label{cor_composed}
Let $\seq{R}$ be a dynamic risk measure on $\LdpF{T}$ and $\seq{A}$ its dynamic acceptance sets. Let $M_t \subseteq M_{t+1}$ for every time $t \in \lrcurly{0,1,...,T-1}$. Then, the following are equivalent:
\begin{enumerate}
\item $ (\tilde{R}_t)_{t=0}^T$ is defined as in \eqref{eqn_composed_final} and \eqref{eqn_composed};
\item $\tilde{A}_{T} =  A_{T}$ and $\tilde{A}_t = \tilde{A}_{t+1}+ A_{t,t+1}^{M_{t+1}}$
for each time $t \in \lrcurly{0,1,...,T-1}$, where $\tilde{A}_{s}$ is the acceptance set of $\tilde{R}_{s}$ for all $s$.
\end{enumerate}
\end{corollary}

\begin{proof}
The proof is analogous to the proof of lemma~\ref{lemma_stepped}, where $R$ and $A$ is replaced by $\widetilde{R}$ and $\widetilde{A}$ at the appropriate places.
\end{proof}

\section{Dual representation}
\label{sec_dual_rep}
In section~\ref{sec_dynamic_defn}, we discussed the primal representation for conditional risk measures.  In this section, we will develop a dual representation by a direct application of the set-valued duality developed by~\cite{H09}.
This representation provides a probability based representation for finding the set of risk compensating portfolios.  In particular, we will demonstrate that, as in the scalar framework, closed convex and coherent risk measures have a representation as the supremum of penalized conditional expectations.

There are multiple approaches that have been used to obtain duality results for scalar dynamic risk measures.  In~\cite{D06,AD07} a dual representation is given as the logical extension of the static case and shown to have the desired properties without directly involving a duality theory for the conditional risk measure.  In~\cite{RS05} an omega-wise approach is used, as it reduces to an omega-wise application of biconjugation of (extended) real-valued functions.  A popular approach, used in~\cite{DS05,FP06,CK10,CDK06,KS05,BN09}, proves the dual representation for dynamic risk measures through a mathematical trick using the static dual representation.  Another approach is given by a direct application of vector-valued duality to the conditional risk measure as in \cite{KPf09,K11} and the first part of \cite{FKV11}. But this approach needs additional strong assumptions (non-emptiness of the subdifferential), and as in the setting of \cite{KPf09,K11}, does not allow for local properties. A further possibility is to use the module approach introduced and applied in \cite{FKV09,FKV11}.
For set-valued dynamic risk measures, \cite{TL12} defines the dual form by the intersection of supporting hyperplanes, but it is not shown how this is related to the traditional scalar dual representation with conditional expectations.  In the present paper, we will apply a new approach, the set-valued approach based on \cite{H09}, which has not been applied to dynamic risk measures so far (but could also be used in the scalar case). This is the most intuitive approach for us as the risk measures under consideration are by nature set-valued functions.

In the scalar framework, most of the papers consider conditionally convex dynamic risk measures (see~\cite{DS05,FP06,CDK06}). In~\cite{KS05,BN04,BN09} the dual representation is deduced for convex and local dynamic risk measures, and in~\cite{CDK06} it was shown that any risk measure on $L^{\infty}$ satisfies the local property.  Using the set valued approach we are able to provide a dual representation for any convex dynamic risk measure for any $p \in [1,+\infty]$, and thus extend, as a byproduct, also the scalar case.

Since we will be considering conjugate duality, we will need to assume from now on that $p \in \lrsquare{1,+\infty}$ and $q$ is such that $\frac{1}{p}+\frac{1}{q} = 1$.  That means we consider the dual pair $\lrparen{\LdpF{T}, \LdqF{T}}$  and endow it with the norm topology, respectively the $\sigma \lrparen{\LdiF{T},\LdoF{T}}$-topology on $\LdiF{T}$ in the case $p = +\infty$.  As before let the set of eligible portfolios $M_t$ be a closed subspace of $\LdpF{t}$ for all times $t$.

Recall that $K_t^{M_t}:=M_t\cap \LdpK{p}{t}{K_t}\subseteq \LdpF{t}$. For all $t \in \{0,1,...,T\}$, we will denote the positive dual cones of $K_t$ and $K_t^{M_t}$ by $\plus{K_t}$ and $\plusp{K_t^{M_t}}$ respectively.  We should note that $\plusp{\LdpK{p}{t}{K_t}} = \LdpK{q}{t}{\plus{K_t}}$, see section 6.3 in \cite{HHR10}.  It holds
\begin{align*}
\plusp{K_t^{M_t}} & = \lrcurly{v \in \LdqF{t}: \forall u \in K_t^{M_t}: \E{\trans{v}u} \geq 0} = \lrparen{\LdpK{q}{t}{\plus{K_t}} + \prp{M_t}} \subseteq \LdqF{t},
\end{align*}
where
$\prp{M_t} = \lrcurly{v \in \LdqF{t}: \forall u \in M_t: \E{\trans{v}u} = 0}$.
Denote \[\plusp{\lrparen{M_t}_+}=\lrcurly{v \in \LdqF{t}: \forall u \in \lrparen{M_t}_+: \E{\trans{v}u} \geq 0}.\]
Recall from remark~\ref{rem image space} that $R_t(X) = R_t(X) + \lrparen{M_t}_+$ by translativity and monotonicity and that a closed convex conditional risk measure $R_t$ maps into the set
$\mathcal{G}(M_t;(M_t)_+)= \lrcurly{D \subseteq M_t: D = \operatorname{cl}\operatorname{co}\lrparen{D + \lrparen{M_t}_+}}$.
Let us denote the positive halfspace with respect to $v \in \LdqF{t}$ by \[G_t(v) = \lrcurly{x \in \LdpF{t}: 0 \leq \E{\trans{v}x}}.\]

\subsection{Set-valued duality}
\label{sec_dual_theory}
Before we begin describing the dual representation for dynamic risk measures we will give a quick introduction and review to set-valued duality results developed by~\cite{H09}.  For the duration of this subsection we will consider the topological dual pairs $(X,X^*)$ and $(Z,Z^*)$ with the (partial) ordering on $Z$ defined by the cone $C \subseteq Z$.  Let us consider the set-valued function $f: X \to \mathcal{G}(Z;C)$.

Set-valued duality uses two dual variables, which we will denote $x^* \in X^*$ and $z^* \in \plus{C} = \lrcurly{z^* \in Z^*: \forall z \in C: \langle z^*,z \rangle \geq 0}$.  The first dual variable, $x^*$ behaves the same as the dual element in the scalar case.  The second dual variable, $z^*$, reflects the order relation in the image space as it is an element in the positive dual cone of the ordering cone.

Using these dual variables one can define a class of set-valued functions that will serve as a (set-valued) replacement for continuous linear functions used in the scalar duality theory, and continuous linear operators used in vector-valued theories, as follows. Let $2^Z$ denote the power set of $Z$, including the empty set.
\begin{definition}
\label{deflin}[\cite{H09} example 2 and proposition 6]
Given $x^* \in X^*$ and $z^* \in Z^*$, the function $F_{(x^*,z^*)}^Z: X \to 2^Z$ is defined for any $x \in X$ by \[F_{(x^*,z^*)}^Z(x) := \lrcurly{z \in Z: \langle x^*,x \rangle \leq \langle z^*,z \rangle}.\]
\end{definition}

Before we can define the convex conjugate, we must define the set-valued infimum and supremum.  As in~\cite{H09,L11}, the set-valued infimum and supremum on $\mathcal{G}(Z;C)$ are given by
\begin{align*}
\inf_{x \in A} f(x) := \cl \co \bigcup_{x \in A} f(x);\quad\quad\quad\quad
\sup_{x \in A} f(x)  := \bigcap_{x \in A} f(x)
\end{align*}
for any $A \subseteq X$ and $f: X \to \mathcal{G}(Z;C)$.
In this way we can then define the (negative) set-valued convex conjugate as
\begin{equation}
-f^*(x^*,z^*) = \inf_{x \in X} \lrsquare{F_{(x^*,z^*)}^Z(-x) + f(x)} = \cl \co \bigcup_{x \in X} \lrsquare{F_{(x^*,z^*)}^Z(-x) + f(x)}
\end{equation}
and the set-valued biconjugate as
\begin{align}
\nonumber f^{**}(x) & = \sup_{(x^*,z^*) \in X^* \times \plus{C} \backslash \{0\}} \lrsquare{-f^*(x^*,z^*) + F_{(x^*,z^*)}^Z(x)}\\
& = \bigcap_{(x^*,z^*) \in X^* \times \plus{C} \backslash \{0\}} \lrsquare{-f^*(x^*,z^*) + F_{(x^*,z^*)}^Z(x)}.
\end{align}
From this form it is clear to see how the set-valued functions $F_{(x^*,z^*)}^Z$ replace the linear functionals in convex analysis in the scalar framework.

\begin{remark}
If $f: X \to \mathcal{G}(Z;C)$ is convex then the (negative) conjugate is equivalent to \[-f^*(x^*,z^*) = \cl \bigcup_{x \in X} \lrsquare{F_{(x^*,z^*)}^Z(-x) + f(x)}.\]
\end{remark}

Finally, the Fenchel-Moreau theorem for set-valued functions reads as follows.
\begin{theorem}
\label{Fenchel-Moreau}[Theorem 2 in \cite{H09}]
A proper function $f: X \to \mathcal{G}(Z;C)$ (i.e. $f(x) \neq Z$ for every $x \in X$ and $f(x) \neq \emptyset$ for some $x \in X$) is closed and convex if and only if $f(x) = f^{**}(x)$ for every $x \in X$.
\end{theorem}

\subsection{Dual variables}
\label{subsec dual variables}
Turning now back to risk measures, we will construct the biconjugate for closed convex and coherent risk measures.  As noted in Fact 3 of section 6.3 in~\cite{HHR10}, any closed convex risk measure is a proper function by finiteness at zero.  Therefore by theorem~\ref{Fenchel-Moreau}, any closed convex or coherent risk measure is equivalent to its biconjugate, which we refer to as the dual representation.

The set-valued duality, given in subsection~\ref{sec_dual_theory}, greatly reduces the work for finding the dual representation for dynamic risk measures as compared to the scalar framework.  This is because the set-valued duality theory works with the same type of image space that (dynamic) set-valued risk measures map into (see remark~\ref{rem image space}), and not necessarily just the extended reals.

In contrast to the static risk measure case discussed in \cite{HH10,HHR10}, we need to consider functions mapping into the power set of $\LdpF{t}$ and thus will generalize definition 3.1 and proposition 3.2 in \cite{HHR10} to this more general case.

Then, the set-valued functionals of definition~\ref{deflin} (and as discussed in subsection~\ref{sec_dual_theory}) are given as follows
\begin{definition}
\label{defn_FYv}
Given $Y \in \LdqF{T}$, $v \in \LdqF{t}$, then the function $\FYvblank: \LdpF{T} \to 2^{M_t}$  is defined by
\begin{equation*}
\FYv{X} = \lrcurly{u \in M_t: \E{\trans{X}Y} \leq \E{\trans{v}u}}.
\end{equation*}
\end{definition}

In the following proposition we consider the relation between properties of these functionals and conditions on the sets of dual variables.
\begin{proposition}
\label{propo properties FYv}
Let $R_t(X) = \FYv{-X}$ for some $Y \in \LdqF{T}$, $v \in \LdqF{t}$, then $R_t$:
\begin{enumerate}
\item is additive and positive homogeneous with $\FYv{0} = G_t(v) \cap M_t = \lrcurly{x \in M_t: 0 \leq \E{\trans{v}x}}$,
\item has a closed graph, and hence closed values, namely closed half spaces,
\item is finite at 0 if and only if it is finite everywhere if and only if $v \in \LdqF{t} \backslash \prp{M_t}$, moreover $R_t(X) \in \lrcurly{M_t, \emptyset}$ for all $X \in \LdpF{T}$ if and only if $v \in \prp{M_t}$,
\item satisfies $R_t(X) = R_t(X) + \lrparen{M_t}_+$ for all $X \in \LdpF{T}$ if and only if $v \in \plusp{\lrparen{M_t}_+}$,
\item is $\LdpF{T}_+$-monotone if and only if $Y \in \LdqF{T}_+$,
\item is $M_t$-translative if and only if $v \in \Et{Y}{t} + \prp{M_t}$,
\item is $K_t$-compatible if and only if $v \in \plusp{K_t^{M_t}}$, and
\item has the corresponding acceptance set given by
\begin{equation*}
A_{R_t} = \lrcurly{X \in \LdpF{T}: 0 \leq \E{\trans{Y}X}}.
\end{equation*}
\end{enumerate}
\end{proposition}
\begin{proof}
This is an adaption of proposition~3.2 in \cite{HHR10} by using example~2 and proposition~6 in \cite{H09} with the linear space $Z$ chosen to be $\LdpF{t}$.
\end{proof}

\begin{remark}
\label{rem_Mplus}
Just as translativity and monotonicity imply that $R_t(\cdot) = R_t(\cdot) + \lrparen{M_t}_+$, it can be seen that $\plusp{\lrparen{M_t}_+} \supseteq \Et{Y}{t} + \prp{M_t}$ if $Y \in \LdqF{T}_+$.
\end{remark}

\begin{remark}
The functions $X \mapsto \FYv{X}, \FYv{-X}$ map into the collection $\mathcal{G}(M_t;(M_t)_+)$ if and only if $v \in \plusp{\lrparen{M_t}_+}$.  By remark~\ref{rem_Mplus}, the image space is $\mathcal{G}(M_t;(M_t)_+)$ and the functions are finite at 0 if the dual variables $(Y,v) \in \LdqF{T}_+ \times \lrsquare{\lrparen{\Et{Y}{t} + \prp{M_t}} \backslash \prp{M_t}}$.
\end{remark}

\begin{remark}
The functions $X \mapsto \FYv{X}, \FYv{-X}$ map into the collection $\mathcal{G}(M_t;K_t^{M_t})$ if and only if $v \in \plusp{K_t^{M_t}}$. This can be concluded because of $K_t$-compatibility.
\end{remark}

Using set-valued biconjugation as discussed in section~\ref{sec_dual_theory} it is possible to give a dual representation for closed convex risk measures already.  However, with the dual representation for scalar dynamic risk measures in mind, we would expect the conditional expectations to appear in the dual representation (i.e. also in definition~\ref{defn_FYv}) for set-valued dynamic risk measures.  We accomplish this by transforming the classical dual variables $(Y,v)$, appearing above, into dual pairs involving vector probability measures $\Q \in \mathcal{M}_d(\P)$.  We denote by $\mathcal{M}_d(\P) := \mathcal{M}_d(\mathbb{P})\lrparen{\Omega,\Ft{T}}$ the set of all vector probability measures with components being absolutely continuous with respect to $\mathbb{P}$.  That is, $\mathbb{Q}_i: \Ft{T} \to [0,1]$ is a probability measure on $\lrparen{\Omega,\Ft{T}}$ such that for all $i \in \{1,2,...,d\}$ we have $\dQidP \in \LoF{T}$.  Using this transformation, defined below in lemma~\ref{lemma_YvQw}, we demonstrate a clear comparison to the dual representation of conditional risk measures in the scalar framework, as given in \cite{AP10,FS11} for example.

For notational purposes, for the rest of the paper we will denote $\diag{w}$ to be the diagonal matrix with the elements of a vector $w$ as the main diagonal.

We will use a $\P$-almost sure version of the $\Q$-conditional expectation defined as follows, see e.g. \cite{CK10}.
Let the sets of one-step transition densities be given by
\[\mathcal{D}_t := \lrcurly{\xi \in \LoF{t}_+: \Et{\xi}{t-1} = 1}\]
for any $t = 1,...,T$.  Then for any $\Q_i \ll \P$ there exists a sequence $(\xi_1^i,...,\xi_T^i) \in \mathcal{D}_1 \times \cdots \times \mathcal{D}_T$ such that $\dQidP = \xi_1^i \cdots \xi_T^i$ by defining
\[\xi_t^i(\omega) := \begin{cases}\frac{\Et{\dQidP}{t}(\omega)}{\Et{\dQidP}{t-1}(\omega)} & \text{if } \Et{\dQidP}{t-1}(\omega) > 0\\ 1 & \text{else} \end{cases}\]
for any $\omega \in \Omega$ and any $t = 1,...,T$.
Conversely, if given a sequence $(\xi_1^i,...,\xi_T^i) \in \mathcal{D}_1 \times \cdots \times \mathcal{D}_T$ then there exists a $\Q_i \ll \P$ such that $\dQidP = \xi_1^i \cdots \xi_T^i$.

Hereafter we will use the convention that for any $\Q \in \mathcal{M}_d(\P)$ and any $t,\tau \in \{0,1,...,T-1\}$ such that $t < \tau$
\[\diag{\Et{\dQdP}{t}}^{-1} \Et{\dQdP}{\tau} := \diag{\xi_{t+1}} \cdots \diag{\xi_{\tau-1}} \xi_{\tau}\]
where $\xi_{s} = \transp{\xi_{s}^1,...,\xi_{s}^d}$ for any time $s$ with $\dQidP = \xi_1^i \cdots \xi_T^i$.  In this way the conditional expectation $\EQt{X}{t}$ is defined $\P$-almost surely as $\EQt{X}{t} := \diag{\Et{\dQdP}{t}}^{-1} \Et{\diag{\dQdP}X}{t}$.

In lemma~\ref{lemma_YvQw} below, a one-to-one correspondence between the dual variables $(Y,v)$ from set-valued duality theory and dual variables based on probability measures $(\Q,w)$ is established.  Then with the probability measure based dual variables, we see that the values of the set-valued functionals $\FQwblank$ are half spaces shifted by the $\Q$-conditional expectation.

\begin{lemma}
\label{lemma_YvQw}
\begin{enumerate}
\item Let $Y \in \LdqF{T}_+$ and $v \in \lrparen{\Et{Y}{t} + \prp{M_t}} \backslash \prp{M_t}$, thus we assume $X\mapsto\FYvblank[-X]$ of definition~\ref{defn_FYv} to be $M_t$-translative, $\LdpF{T}_+$-monotone and to be finite at 0.  Then there exists a $\mathbb{Q} \in \mathcal{M}_d(\mathbb{P})$ and a $w \in \plusp{\lrparen{M_t}_+} \backslash \prp{M_t}$ such that \[\diag{w}\diag{\Et{\dQdP}{t}}^{-1}\dQdP \in \LdqF{T}_+\]  and $\FYvblank = \FQwblank$, where
    \begin{equation}
    \label{FQw}
    \FQw{X} =  \lrcurly{u \in M_t: \E{\trans{w}\EQt{X}{t}} \leq \E{\trans{w}u}} = \lrparen{\EQt{X}{t} + G_t(w)} \cap M_t.
    \end{equation}
\item Vice versa, let $\mathbb{Q} \in \mathcal{M}_d(\mathbb{P})$ and $w \in \plusp{\lrparen{M_t}_+} \backslash \prp{M_t}$  such that the relationship \[\diag{w}\diag{\Et{\dQdP}{t}}^{-1}\dQdP \in \LdqF{T}_+\] holds. Then there exists a $Y \in \LdqF{T}_+$ and $v \in \lrparen{\Et{Y}{t} + \prp{M_t}} \backslash \prp{M_t}$ such that $\FYvblank = \FQwblank$.
\end{enumerate}
\end{lemma}
\begin{proof}
(i) Let $w = \Et{Y}{t}$ then $w \in \LdqF{t}_+$.  Then since $v \in \lrparen{\Et{Y}{t} + \prp{M_t}} \backslash \prp{M_t}$ we have $v \in w + \prp{M_t}$ or equivalently $w \in v + \prp{M_t}$.  Additionally we have, by remark~\ref{rem_Mplus}, that $v \in \plusp{\lrparen{M_t}_+}$, therefore $w \in \plusp{\lrparen{M_t}_+} + \prp{M_t}$.  And because $v \notin \prp{M_t}$ we have $w \notin \prp{M_t}$, therefore $w \in \plusp{\lrparen{M_t}_+} \backslash \prp{M_t} + \prp{M_t}$.  It can easily be seen that the set $\plusp{\lrparen{M_t}_+} \backslash \prp{M_t} + \prp{M_t}$ is equal to $\plusp{\lrparen{M_t}_+} \backslash \prp{M_t}$.

    Let $\dQidP = \frac{Y_i}{\E{Y_i}}$ if $\E{Y_i} > 0$ and arbitrarily in $\LdqF{T}_+$ such that $\E{\dQidP} = 1$ if $\E{Y_i} = 0$.  Then $Y = \diag{w}\diag{\Et{\dQdP}{t}}^{-1}\dQdP \in \LdqF{T}_+$.
  From the above, we can conclude that $\FYvblank = \FQwblank$ by
    \begin{equation*}
    \E{\trans{X}Y} = \E{\trans{X}\diag{w}\diag{\Et{\dQdP}{t}}^{-1}\dQdP} = \E{\trans{w}\EQt{X}{t}}
    \end{equation*}
    and $\E{\trans{v}u} = \E{\trans{w}u}$ for $u \in M_t$ since $w \in v + \prp{M_t}$.

(ii) Let $Y = \diag{w}\diag{\Et{\dQdP}{t}}^{-1}\dQdP \in \LdqF{T}_+$ then trivially we have \[\Et{Y}{t} = \Et{\diag{w}\diag{\Et{\dQdP}{t}}^{-1}\dQdP}{t} = w\]  and
    $\E{\trans{X}Y} = \E{\trans{X}\diag{w}\diag{\Et{\dQdP}{t}}^{-1}\dQdP}$.
  From the assumption and $\plusp{\lrparen{M_t}_+} \backslash \prp{M_t} = \plusp{\lrparen{M_t}_+} \backslash \prp{M_t} + \prp{M_t}$ it holds $w \in \plusp{\lrparen{M_t}_+} + \prp{M_t}$.  Thus, $w = w_{\plusp{\lrparen{M_t}_+}} + w_{\prp{M_t}}$ where $w_{\plusp{\lrparen{M_t}_+}} \in \plusp{\lrparen{M_t}_+}$ and $w_{\prp{M_t}} \in \prp{M_t}$.  In particular, $w_{\plusp{\lrparen{M_t}_+}} = w - w_{\prp{M_t}} \in \Et{Y}{t} + \prp{M_t} \subseteq \LdqF{t}$.  Set $v = w_{\plusp{\lrparen{M_t}_+}}$.  Furthermore, $w \notin \prp{M_t}$ implies $v \notin \prp{M_t}$.  Thus $v \in \lrparen{\Et{Y}{t} + \prp{M_t}} \backslash \prp{M_t}$.  We have $\E{\trans{w}u} = \E{\trans{v}u}$ for every $u \in M_t$ since $w \in v + \prp{M_t}$.
\end{proof}

\subsection{Convex and coherent risk measures}
Utilizing set-valued duality, proposed in~\cite{H09}, and the transformation of the dual variables as described by lemma~\ref{lemma_YvQw}, we can now give the dual representation for set-valued closed convex and coherent dynamic risk measures.  As we demonstrated in section~\ref{subsec dual variables}, the set of dual variables can be defined by
\begin{align*}
\mathcal{W}_t^q & = \lcurly{(\mathbb{Q},w) \in \mathcal{M}_d(\mathbb{P}) \times \lrparen{\plusp{\lrparen{M_t}_+} \backslash \prp{M_t}}:}\; \rcurly{\diag{w}\diag{\Et{\dQdP}{t}}^{-1}\dQdP \in \LdqF{T}_+}.
\end{align*}
In this way we have two dual variables, the first is a (vector) probability measure and the second contains the ordering of the eligible portfolios.  The additional coupling condition of a pair of dual variables $(\Q,w)$ guarantees the $\Q$-conditional expectation of a ($\P$-a.s.) greater portfolio is dominant in the ordering defined by $w$ as well.

\begin{definition}
\label{def penalty}
A function $-\alpha_t: \mathcal{W}_t^q \to \mathcal{G}(M_t;(M_t)_+)$ is a \textbf{\emph{penalty function}} at time $t$ if it satisfies
\begin{enumerate}
\item $\cap_{(\mathbb{Q},w) \in \mathcal{W}_t^q} -\alpha_t(\mathbb{Q},w) \neq \emptyset$ and $-\alpha_t(\mathbb{Q},w) \neq M_t$ for at least one $(\mathbb{Q},w) \in \mathcal{W}_t^q$ and
\item $-\alpha_t(\mathbb{Q},w) = \operatorname{cl}\lrparen{-\alpha_t(\mathbb{Q},w) + G_t(w)} \cap M_t$ for all $(\mathbb{Q},w) \in \mathcal{W}_t^q$.
\end{enumerate}
\end{definition}

Then, the duality theorem 4.2 from \cite{HHR10} extends to the dynamic case in the following way.
\begin{theorem}
\label{thm_dual}
A function $R_t: \LdpF{T} \to \mathcal{G}(M_t;(M_t)_+)$ is a closed \textbf{\emph{convex conditional risk measure}} if and only if there is a penalty function $-\alpha_t$ at time $t$ such that
\begin{equation}
\label{convex_dual}
R_t(X) = \bigcap_{(\mathbb{Q},w) \in \mathcal{W}_t^q} \lrsquare{-\alpha_t(\mathbb{Q},w) + \lrparen{\EQt{-X}{t} + G_t\lrparen{w}} \cap M_t}.
\end{equation}
In particular, for $R_t$ with the aforementioned properties, equation~\eqref{convex_dual} is satisfied with the minimal penalty function $-\alpha_t^{\min}$ defined by
\begin{equation}
\label{min penalty}
-\alpha_t^{\min}(\mathbb{Q},w) = \operatorname{cl}\bigcup_{Z \in A_{R_t}} \lrparen{\EQt{Z}{t} + G_t(w)} \cap M_t.
\end{equation}
The penalty function $-\alpha_t^{\min}$ has the property that for any penalty function $-\alpha_t$ satisfying equation~\eqref{convex_dual} it holds that $-\alpha_t(\mathbb{Q},w) \supseteq -\alpha_t^{\min}(\mathbb{Q},w)$ for all $(\mathbb{Q},w) \in \mathcal{W}_t^q$.
\end{theorem}
\begin{proof}
This follows from theorem 2 in \cite{H09} by applying lemma 3.6 and remark 3.7 in a similar way as it was done in theorem 4.2 in \cite{HHR10}, the proof of which is given in section 6.3 of that same paper.
\end{proof}

\begin{corollary}
\label{cor_dual}
A function $R_t: \LdpF{T} \to \mathcal{G}(M_t;(M_t)_+)$ is a closed \textbf{\emph{coherent conditional risk measure}} if and only if there is a nonempty set $\mathcal{W}_{t,R_t}^q \subseteq \mathcal{W}_t^q$ such that
\begin{equation}
\label{coherent_dual}
R_t(X) = \bigcap_{(\mathbb{Q},w) \in \mathcal{W}_{t,R_t}^q} \lrparen{\EQt{-X}{t} + G_t\lrparen{w}} \cap M_t.
\end{equation}
In particular, equation~\eqref{coherent_dual} is satisfied with $\mathcal{W}_{t,R_t}^q$ replaced by $\mathcal{W}_{t,R_t}^{q,\max}$ with
\begin{equation}
\label{max_dualset}
\mathcal{W}_{t,R_t}^{q,\max} = \lrcurly{(\mathbb{Q},w) \in \mathcal{W}_t^q: \diag{w}\diag{\Et{\dQdP}{t}}^{-1}\dQdP \in \plus{A_{R_t}}}
\end{equation}
and if $\mathcal{W}_{t,R_t}^q$ satisfies equation~\eqref{coherent_dual} then the inclusion $\mathcal{W}_{t,R_t}^q \subseteq \mathcal{W}_{t,R_t}^{q,\max}$ holds.
\end{corollary}
\begin{proof}
This corollary follows from theorem 3.11 and an adaption of equation (6.4) in proposition 6.7 in \cite{HHR10} to the dynamic case.
\end{proof}

Let us turn to the special case of conditional convexity and coherence, which is the usual property imposed on dynamic risk measures in the scalar case.  The corresponding duality theorem in the set-valued case is given in corollary~\ref{cor_conditional_dual} below. It will turn our that the penalty functions will have the following additional property in this case.

Consider a penalty function $-\alpha_t: \mathcal{W}_t^q \to \mathcal{G}(M_t;(M_t)_+)$ at time $t$ such that for any $(\mathbb{Q},w) \in \mathcal{W}_t^q$ and any $\lambda \in \LiF{t} \st 0 < \lambda < 1$ it holds
\begin{equation}
\label{conditional_alpha}
-\alpha_t(\mathbb{Q},w) \supseteq \lambda(-\alpha_t(\mathbb{Q},\lambda w)) + (1-\lambda)(-\alpha_t(\mathbb{Q},(1-\lambda) w)).
\end{equation}

\begin{corollary}
\label{cor_conditional_dual}
A function $R_t: \LdpF{T} \to \mathcal{G}(M_t;(M_t)_+)$ is a closed \textbf{\emph{conditionally convex}} conditional risk measure if and only if there is a penalty function $-\alpha_t$ at time $t$ satisfying \eqref{conditional_alpha} such that equation~\eqref{convex_dual} holds true.  In particular, for $R_t$ with the aforementioned properties, the minimal penalty function $-\alpha_t^{\min}$, defined in equation~\eqref{min penalty}, satisfies \eqref{conditional_alpha}.

Further, a function $R_t: \LdpF{T} \to \mathcal{G}(M_t;(M_t)_+)$ is a closed \textbf{\emph{conditionally coherent}} conditional risk measure if and only if there is a nonempty set $\mathcal{W}_{t,R_t}^q \subseteq \mathcal{W}_t^q$ conditionally conical in the second variable, i.e. for any $(\mathbb{Q},w) \in \mathcal{W}_{t,R_t}^q$ and $\lambda \in \LiF{t}_{++}$ then $(\mathbb{Q},\lambda w) \in \mathcal{W}_{t,R_t}^q$, such that equation~\eqref{coherent_dual} is satisfied.  In particular, for $R_t$ with the aforementioned properties, the maximal dual set $\mathcal{W}_{t,R_t}^{q,\max}$, defined in equation~\eqref{max_dualset}, satisfies this additional condition.
\end{corollary}

\begin{proof}
Using  theorem~\ref{thm_dual}, only two things remain to show: First, if $-\alpha_t$
is a  penalty function satisfying \eqref{conditional_alpha}, then the risk measure defined by \eqref{convex_dual} is conditionally convex. Second,
the minimal penalty function $-\alpha_t^{\min}$ of a conditionally convex risk measure satisfies \eqref{conditional_alpha}.

 First, if $-\alpha_t$ is a penalty function satisfying \eqref{conditional_alpha}, then for any $\lambda \in \LiF{t} \st 0 < \lambda < 1$
\begin{align*}
R_t(\lambda X + (1-\lambda)Y) &= \bigcap_{(\mathbb{Q},w) \in \mathcal{W}_t^q} \lrsquare{-\alpha_t(\mathbb{Q},w) + \lrparen{\EQt{-(\lambda X + (1-\lambda)Y)}{t} + G_t\lrparen{w}} \cap M_t}\\
&= \bigcap_{(\mathbb{Q},w) \in \mathcal{W}_t^q} \lrsquare{-\alpha_t(\mathbb{Q},w) + \lrparen{\lambda\EQt{-X}{t} + (1-\lambda)\EQt{-Y}{t} + G_t\lrparen{w}} \cap M_t}\\
&\supseteq \bigcap_{(\mathbb{Q},w) \in \mathcal{W}_t^q} \lsquare{-\alpha_t(\mathbb{Q},w) + \lrparen{\lambda\EQt{-X}{t} + G_t(w)} \cap M_t}\\
&\quad\quad \rsquare{+ \lrparen{(1-\lambda)\EQt{-Y}{t} + G_t\lrparen{w}} \cap M_t}\\
&= \bigcap_{(\mathbb{Q},w) \in \mathcal{W}_t^q} \lsquare{-\alpha_t(\mathbb{Q},w) + \lambda \lrparen{\EQt{-X}{t} + G_t(\lambda w)} \cap M_t}\\
&\quad\quad \rsquare{+ (1-\lambda)\lrparen{\EQt{-Y}{t} + G_t\lrparen{(1-\lambda) w}} \cap M_t}
\end{align*}
\begin{align*}
&\supseteq \bigcap_{(\mathbb{Q},w) \in \mathcal{W}_t^q} \lsquare{\lambda(-\alpha_t(\mathbb{Q},\lambda w)) + (1-\lambda)(-\alpha_t(\mathbb{Q},(1-\lambda) w))}\\
&\quad\quad \rsquare{+ \lambda \lrparen{\EQt{-X}{t} + G_t(\lambda w)} \cap M_t + (1-\lambda)\lrparen{\EQt{-Y}{t} + G_t\lrparen{(1-\lambda) w}} \cap M_t}\\
&\supseteq \bigcap_{(\mathbb{Q},w) \in \mathcal{W}_t^q} \lrsquare{\lambda(-\alpha_t(\mathbb{Q},\lambda w)) + \lambda \lrparen{\EQt{-X}{t} + G_t(\lambda w)} \cap M_t}\\
&\quad\quad + \bigcap_{(\mathbb{Q},w) \in \mathcal{W}_t^q} \lrsquare{(1-\lambda)(-\alpha_t(\mathbb{Q},(1-\lambda) w)) + (1-\lambda)\lrparen{\EQt{-Y}{t} + G_t\lrparen{(1-\lambda) w}} \cap M_t}\\
&= \lambda \bigcap_{(\mathbb{Q},w) \in \mathcal{W}_t^q} \lrsquare{-\alpha_t(\mathbb{Q},\lambda w) + \lrparen{\EQt{-X}{t} + G_t(\lambda w)} \cap M_t}\\
&\quad\quad + (1-\lambda) \bigcap_{(\mathbb{Q},w) \in \mathcal{W}_t^q} \lrsquare{(-\alpha_t(\mathbb{Q},(1-\lambda) w)) + \lrparen{\EQt{-Y}{t} + G_t\lrparen{(1-\lambda) w}} \cap M_t}\\
&\supseteq \lambda R_t(X) + (1-\lambda) R_t(Y).
\end{align*}
The last line above follows since if $(\mathbb{Q},w) \in \mathcal{W}_t^q$ then $(\mathbb{Q},\lambda w) \in \mathcal{W}_t^q$.  The conditional convexity of $R_t$ can be extended to $\lambda \in \LiF{t} \st 0 \leq \lambda \leq 1$ by taking a sequence $(\lambda_n)_{n = 0}^{\infty} \subseteq \LiF{t}$ such that $0 < \lambda_n < 1$ for every $n \in \mathbb{N}$ which converges almost surely to $\lambda$.  Then by dominated convergence $\lambda_n X$ converges to $\lambda X$ in the norm topology ($\sigma(\LdiF{T},\LdoF{T})$ topology if $p = \infty$).  Therefore, for any $X,Y \in \LdpF{T}$
\begin{align*}
R_t(\lambda X + (1-\lambda) Y) &= R_t(\lim_{n \to \infty} (\lambda_n X + (1-\lambda_n) Y))\\
&\supseteq \liminf_{n \to \infty} R_t(\lambda_n X + (1-\lambda_n) Y)\\
&\supseteq \liminf_{n \to \infty} [\lambda_n R_t(X) + (1-\lambda_n) R_t(Y)]\\
&\supseteq \lambda R_t(X) + (1-\lambda) R_t(Y)
\end{align*}
by $R_t$ closed (see proposition~2.34 in \cite{L11}) and conditionally convex on the interval $0 < \lambda_n < 1$. Note that we use the convention from~\cite{L11} that the limit inferior of a net of sets $(B_i)_{i \in I}$ is given by $\liminf_{i \in I} B_i = \bigcap_{i\in I} \cl\bigcup_{j \geq i} B_i$.

Conversely, let $R_t$ be a conditionally convex risk measure, then its acceptance set $A_{R_t}$ is conditionally convex as well. Therefore for any $(\mathbb{Q},w) \in \mathcal{W}_t^q$ and any $\lambda \in \LiF{t} \st 0 < \lambda< 1$
\begin{align*}
-\alpha_t^{\min}(\mathbb{Q},w) &= \operatorname{cl}\bigcup_{Z \in A_{R_t}} \lrparen{\EQt{Z}{t} + G_t(w)} \cap M_t\\
&\supseteq \operatorname{cl}\bigcup_{Z \in \lambda A_{R_t} + (1-\lambda) A_{R_t}} \lrparen{\EQt{Z}{t} + G_t(w)} \cap M_t\\
&= \operatorname{cl}\bigcup_{Z_1,Z_2 \in A_{R_t}} \lrparen{\EQt{\lambda Z_1 + (1-\lambda) Z_2}{t} + G_t(w)} \cap M_t\\
&\supseteq \operatorname{cl}\bigcup_{Z \in A_{R_t}} \lambda \lrparen{\EQt{Z}{t} + G_t(\lambda w)} \cap M_t + \operatorname{cl}\bigcup_{Z \in A_{R_t}} (1-\lambda)\lrparen{\EQt{Z}{t} + G_t((1-\lambda) w)} \cap M_t\\
&= \lambda(-\alpha_t^{\min}(\mathbb{Q},\lambda w)) + (1-\lambda)(-\alpha_t^{\min}(\mathbb{Q},(1-\lambda) w)).
\end{align*}
Thus by theorem~\ref{thm_dual} there exists a penalty function with property \eqref{conditional_alpha} satisfying equation~\eqref{convex_dual}.

The proof for the conditionally coherent case follows analogously.
\end{proof}

\begin{example}[{[Example~\ref{ex_worst_case} continued]} ]
As the worst cost risk measure is a closed coherent risk measure with acceptance set $A_t = \LdpF{T}_+$ then the maximal set of dual variables is given by \[\mathcal{W}_{t,R_t^{WC}}^{q,\max} := \lrcurly{(\mathbb{Q},w) \in \mathcal{W}_t^q: \diag{w}\diag{\Et{\dQdP}{t}}^{-1}\dQdP \in \LdqF{T}_+} = \mathcal{W}_t^q.\]
\end{example}

If a dynamic risk measure is additionally market-compatible, i.e. trading opportunities are taken into account for all $t$, then the dual representation can be given with respect to $\mathcal{W}_{t,K}^q \subseteq \mathcal{W}_t^q$.  We will define the set of $K_t$-compatible dual variables by
\begin{align*}
\mathcal{W}_{t,K}^q & = \lcurly{(\mathbb{Q},w) \in \mathcal{M}_d(\mathbb{P}) \times \lrparen{\plusp{K_t^{M_t}} \backslash \prp{M_t}}:}\\
&\quad\rcurly{\diag{w}\diag{\Et{\dQdP}{t}}^{-1}\dQdP \in \LdqF{T}_+}\subseteq \mathcal{W}_t^q.
\end{align*}

In fact, the dual representation for a market-compatible risk measure requires a simple switch in the ordering cone from $(M_t)_+$ to $K_t^{M_t}$ (with corresponding image space $\mathcal{G}\lrparen{M_t;K_t^{M_t}}$ as discussed in remark~\ref{rem image space}).

Theorem~\ref{thm_dual}, applied over the set $W_{t,K}^q$ of dual variables, together with an adaption of remark~6.8 in \cite{HHR10} to the dynamic setting leads to a way to generate market compatible dynamic convex risk measure. If given a sequence of convex conditional acceptance sets, or more generally of nonempty convex sets $\hat{A_t} \subseteq \LdpF{T}$ such that
\begin{equation*}
A_t = \operatorname{cl} \lrparen{\hat{A_t} + K_t^{M_t} }
\end{equation*}
satisfies definition~\ref{defn_acceptance} of a conditional acceptance set, then equation~\eqref{convex_dual} with the penalty function $-\alpha_t$, such that
\begin{equation*}
-\alpha_t(\mathbb{Q},w) = \operatorname{cl} \bigcup_{Z \in A_t} \lrparen{\EQt{Z}{t} + G_t(w)} \cap M_t,
\end{equation*}
produces a sequence of closed $K_t$-compatible convex conditional risk measure, thus a market compatible dynamic convex risk measure for the market defined by the sequence of solvency cones $\seq{K}$.

\section{Examples}
\label{sec_examples}

\subsection{Dynamic superhedging}
\label{section_superhedging}

In this section, we define the dynamic extension of the set of superhedging portfolios in markets with proportional transaction costs as presented in~\cite{K99,S04,KS09,HHR10,LR11}.  We further show that the set of superhedging portfolios yields a set-valued market-compatible coherent dynamic risk measure that is multi-portfolio time consistent.  In \cite{LR11} an algorithm for calculating the set of superhedging prices is presented.  That algorithm relies on a successive calculation of superhedging prices backwards in time and leads to a sequence of linear vector optimization problems that can be solved by Benson's algorithm. We show that the recursive form, which is equivalent to multi-portfolio time consistency, leads to and simplifies the proof of the recursive algorithm given in \cite{LR11}. This result gives a hint that the set-valued recursive form of multi-portfolio time consistent risk measures is very useful in practice and might lead to a set-valued analog of Bellman's principle.

Let the random variable $V_t: \Omega \to \mathbb{R}^d$ be a portfolio vector at time $t$ such that the values of $V_t(\omega)$ are in physical units as described in \cite{K99,S04}.  That is, the $i^{th}$ element of $V_t(\omega)$ is the number of asset $i$ in the portfolio in state $\omega\in\mathcal F_t$ at time $t$.  An $\mathbb{R}^d$-valued adapted process $\seq{V}$ is called a self-financing portfolio process for the market given by the solvency cones $\seq{K}$ if
\begin{equation*}
\forall t = 0,1,..., T: V_t - V_{t-1} \in -K_t
\end{equation*}
where $V_{-1} = 0$.

Let $C_{t,T} \subseteq \LdpF{T}$ be the set of $\LdpF{T}$-valued random vectors $V_T: \Omega \to \mathbb{R}^d$ that are the values of a self-financing portfolio process at terminal time $T$ with endowment $0$ at time $t$.  From this definition it follows that $C_{t,T} = \sum_{s = t}^{T} -\LdpK{p}{s}{K_s} $.

An $\mathbb{R}^d_+$-valued adapted process $Z = \seq{Z}$ is called a
consistent pricing process for the market model $\seq{K}$ if $Z$ is a martingale under the physical measure $\mathbb{P}$ and
\begin{equation*}
\forall t \in \lrcurly{0,1,...,T}: Z_t \in \plus{K_t} \backslash \lrcurly{0}.
\end{equation*}

The market is said to satisfy the robust no arbitrage
property (NA$^\text{r}$) if there exists a market process $(\widetilde{K}_t)_{t=0}^T$
satisfying
\begin{equation}
\label{def K tilde}
 K_t \subseteq \widetilde{K}_t \quad \mbox{and} \quad
    K_t \backslash -K_t \subseteq \Int \widetilde{K}_t\quad \P\mbox{-a.s.}
\end{equation}
for all $t\in\{0,1,\dots,T\}$ such that
\[
\widetilde{C}_{0,T} \cap L^0_d(\mathcal{F}_T,\mathbb{R}^d_+) = \lrcurly{0},
\]
where $\widetilde{C}_{0,T}$ is generated by the self-financing portfolio processes with
\[
\forall t\in\{0,1,\dots,T\} \colon\quad V_{t} - V_{t-1} \in -\widetilde{K}_{t} \quad \P\mbox{-a.s.}
\]

The time $t$ version of theorem~4.1 in \cite{S04}, or theorem 5.2 in \cite{HHR10} reads as follows.

\begin{theorem}
\label{thm_NA}
Assume that the market process $\seq{K}$ satisfies the robust no arbitrage condition (NA$^{\text{r}}$), then the following conditions are equivalent for $X \in \LdpF{T}$ and $u \in \LdpF{t}$:
\begin{enumerate}
\item \label{thm_NA_1} $X - u  \in C_{t,T}$, i.e. there exists a self-financing portfolio process $(V_s)_{s = 0}^T$ with $V_s = 0$ if $s < t$, and $V_s \in \LdpF{s}$ for each time $s \geq t$ such that
\begin{equation}
\label{thm_NA_selffinance}
u + V_T = X.
\end{equation}
\item \label{thm_NA_2} For every consistent pricing process $\seq{Z}$ with $Z_t \in \LdqF{t}$ for each time $t$, it holds that
\begin{equation*}
\E{\trans{X}Z_T} \leq \E{\trans{u}Z_t}.
\end{equation*}
\end{enumerate}
\end{theorem}
\begin{proof}
This is a trivial adaptation of theorem~4.1 in \cite{S04}, or theorem 5.2 in \cite{HHR10}.
\end{proof}

Clearly any element $u \in \LdpF{t}$ satisfying equation \eqref{thm_NA_selffinance} is a superhedging portfolio of $X$ at time $t$.  Thus, as an extension to the static case in \cite{HHR10}, the set of superhedging portfolios defines a closed coherent market-compatible dynamic risk measure on $\LdpF{T}$ as described in the corollary below.

\begin{corollary}
\label{cor_shp}
An element $u \in \LdpF{t}$ is a superhedging portfolio at time $t$ for the claim $X \in \LdpF{T}$ if and only if $u \in SHP_t(X)$ with
\begin{equation}
\label{eq_superhedging_acceptance}
SHP_t(X) := \lrcurly{u \in \LdpF{t}: -X + u  \in -C_{t,T}}.
\end{equation}
If the market process $\seq{K}$ satisfies the robust no arbitrage condition (NA$^{\text{r}}$), then $\seq{R}$ defined by $R_t(X):=SHP_t(-X)$ is a closed conditionally coherent market-compatible dynamic risk measure on $\LdpF{T}$ and has the following dual representation
\begin{equation}
\label{eq_superhedging_dual}
SHP_t(X) = \bigcap_{(\mathbb{Q},w) \in \mathcal{W}^q_{\lrcurly{t,...,T}}} \lrparen{\EQt{X}{t} + G_t(w)},
\end{equation}
where $t\in\{0,1,...,T\}$ and
\begin{align*}
\mathcal{W}^q_{\lrcurly{t,...,T}} & = \lcurly{(\Q,w) \in \mathcal{W}_{t,K}^q: \; \forall s \in \lrcurly{t,...,T}}\\
&\quad \rcurly{\diag{w}\diag{\Et{\dQdP}{t}}^{-1}\Et{\dQdP}{s} \in \LdpK{q}{s}{\plus{K_s}}}.
\end{align*}
\end{corollary}
\begin{proof}
Theorem~\ref{thm_NA} condition \ref{thm_NA_1} implies equation \eqref{eq_superhedging_acceptance} immediately.  Setting $M_t = \LdpF{t}$ for all times $t = 0,1,...,T$, and since $-C_{t,T} = \sum_{s = t}^T \LdpK{p}{s}{K_s}$ it follows from $K_s(\omega)$ being a convex cone with $\mathbb{R}^d_+ \subseteq K_s(\omega)$ for all $s \in \lrcurly{t,...,T}$ and for all $\omega \in \Omega$ that the set $-C_{t,T}$ is an acceptance set at time $t$ as it satisfies definition \ref{defn_acceptance}.
This also trivially implies that $-C_{t,T}$ is market-compatible. Furthermore, $-C_{t,T} \subseteq \LdpF{T}$ is a convex cone and closed in $\LdpF{T}$ (follows as in \cite{S04}). Thus, $R_t(X)=SHP_t(-X)$ as defined in equation \eqref{eq_superhedging_acceptance} is by proposition~\ref{propo A_t-R_t} a closed, coherent, and market-compatible conditional risk measure.

By theorem~\ref{thm_NA} condition~\ref{thm_NA_2} the set $SHP_t(X)$ of superhedging portfolios of X at time $t$ described in equation \eqref{eq_superhedging_acceptance} can also be written in the form
\begin{equation}
\label{eq_superhedging_cpp}
SHP_t(X) = \bigcap_{Z \in CPP_t} \lrcurly{u \in \LdpF{t}: \E{\trans{Z_T}X} \leq \E{\trans{Z_t}u}}
\end{equation}
where $CPP_t$ is the set of consistent pricing processes starting at time $t$ such that for all $s \geq t$ $Z_s \in \LdpK{q}{s}{\plus{K_s}} \backslash \lrcurly{0}$.  Equation \eqref{eq_superhedging_cpp} is equivalent to \eqref{eq_superhedging_dual} as there is a one-to-one relationship between the set $CPP_t$ and $\mathcal{W}^q_{\lrcurly{t,...,T}}$:  Given a consistent pricing process $Z$, we can create a pair $(\mathbb{Q},w) \in \mathcal{W}^q_{\lrcurly{t,...,T}}$ by defining $w := \Et{Z_T}{t} = Z_t \in \LdpK{q}{t}{\plus{K_t}} \backslash \lrcurly{0}$ and
\begin{equation*}
\dQidP := \frac{(Z_T)_i}{\E{(Z_T)_i}}.
\end{equation*}
Conversely, a pair $(\mathbb{Q},w) \in \mathcal{W}^q_{\lrcurly{t,...,T}}$ yields a consistent pricing process $Z$ starting from time $t$ by letting $Z_T = \diag{w}\diag{\Et{\dQdP}{t}}^{-1}\dQdP$ and $Z_s = \Et{Z_T}{s}$ for all $s = t,...,T$.

Finally, conditional coherence follows from corollary~\ref{cor_conditional_dual} and the fact that if $(\mathbb{Q},w) \in \mathcal{W}^q_{\lrcurly{t,...,T}}$ and $\lambda \in \LiF{t}_{++}$ then $(\mathbb{Q}, \lambda w) \in \mathcal{W}^q_{\lrcurly{t,...,T}}$ trivially.
\end{proof}

The probability measures $\mathbb{Q}$ with $(\mathbb{Q},w) \in \mathcal{W}^q_{\lrcurly{t,...,T}}$ can be seen as equivalent martingale measures.  Indeed, the component $\mathbb{Q}_i$, $i = 1,...,d$ is a martingale measure if asset $i$ is chosen as num\'{e}raire.

It remains to show that the dynamic superhedging set (as a coherent risk measure) is multi-portfolio time consistent.
\begin{lemma}
\label{lemma_SHPmptc}
Under the (NA$^{\text{r}}$) condition, the set-valued function $R_t(X):=SHP_t(-X)$ defined in corollary~\ref{cor_shp} is a normalized multi-portfolio time consistent dynamic risk measure.
\end{lemma}
\begin{proof}
We have already shown that the acceptance set of $R_t(X)=SHP_t(-X)$ is $ A_t = -C_{t,T} = \sum_{s = t}^T \LdpK{p}{s}{K_s}$.  The one step acceptance set is given by $ A_{t,t+1}= A_t\cap\LdpF{t+1} = \LdpK{p}{t}{K_t}+\lrparen{\sum_{s = t+1}^T \LdpK{p}{s}{K_s}}\cap\LdpF{t+1}$. Since it holds
\[
\sum_{s = t+2}^T \LdpK{p}{s}{K_s}+\lrparen{\sum_{s = t+2}^T \LdpK{p}{s}{K_s}}\cap\LdpF{t+1}=\sum_{s = t+2}^T \LdpK{p}{s}{K_s},
\]
it can easily be seen that $A_t =  A_{t,t+1} +  A_{t+1}$ is satisfied for any time $t$.  Since $M_t = \LdpF{t} \subseteq \LdpF{t+1} = M_{t+1}$ for all times $t$, theorem \ref{thm_equiv_tc} implies that $\seq{R}$ is a multi-portfolio time consistent dynamic risk measure if it is normalized.  Note that $SHP_t(0)= A_t\cap\LdpF{t}$ and since the solvency cones contain $\mathbb{R}^d_+$, it holds $SHP_t(0) \supseteq \LdpF{t}_+$, and by (NA$^{\text{r}}$) we have $SHP_t(0) \cap \LdpF{t}_{--} = \emptyset$ for every time $t$, which implies for coherent risk measures that $\seq{R}$ is normalized (as mentioned in section~\ref{sec_dynamic_defn} and shown in property~3.1 in \cite{JMT04}).
\end{proof}

In the frictionless case the no arbitrage condition implies $ A_{t,t+1}= A_t\cap\LdpF{t+1} = \LdpK{p}{t}{K_t}+\LdpK{p}{t+1}{K_{t+1}}$ (see e.g. section~4.2 in \cite{P07}), but in markets with transaction costs this is not necessarily true and the one step acceptance sets are in general $ A_{t,t+1}= A_t\cap\LdpF{t+1} = \LdpK{p}{t}{K_t}+\LdpK{p}{t+1}{K_{t+1}}+\lrparen{\sum_{s = t+2}^T \LdpK{p}{s}{K_s}}\cap\LdpF{t+1}$, i.e. it is possible to benefit from $\mathcal F_{t+1}$-measurable trades carried out at $t\in\{t+2,...,T\}$ without creating a (robust) arbitrage. This also means that $SHP_t$ restricted to the space $\LdpF{t+1}$ is no longer equal to the set of superhedging portfolios at time $t$ allowing trading until time $t+1$ as it is in the frictionless case.

\begin{remark}
Since multi-portfolio time consistency is equivalent to recursiveness (theorem~\ref{thm_equiv_tc}), the set of superhedging portfolios satisfies under (NA$^{\text{r}}$)
\begin{equation}
\label{recursive SHP}
SHP_t(X) = \bigcup_{Z \in SHP_{t+1}(X)} SHP_t(Z)=:SHP_t(SHP_{t+1}(X)).
\end{equation}
The recursiveness leads in a straight forward manner to the recursive algorithm given by theorem~3.1 in \cite{LR11} and thus simplifies the proof of that theorem significantly.
Direct observation or using lemma~\ref{lemma_market_comp} for $\seq{R}$ defined by $R_t(X):=SHP_t(-X)$ being market-compatible (corollary~\ref{cor_shp}), normalized and multi-portfolio time consistent (lemma~\ref{lemma_SHPmptc}), yields
$A_{t+1} = A_{t+1} + \sum_{s = t+1}^{T}\LdpK{p}{s}{K_s}$ for each $t$, and in particular
\[A_{t+1} \supseteq A_{t+1} + \lrparen{\sum_{s = t+1}^{T}\LdpK{p}{s}{K_s}}\cap\LdpF{t+1}.\]
Then, proposition~\ref{propo A_t-R_t} implies $SHP_{t+1}(X)=SHP_{t+1}(X)+\lrparen{\sum_{s = t+1}^{T}\LdpK{p}{s}{K_s}}\cap\LdpF{t+1}$ which leads to
$SHP_{t+1}(X)+A_{t,t+1}=SHP_{t+1}(X)+\LdpK{p}{t}{K_t}$. Thus, the recursive form \eqref{recursive SHP} reads as
\begin{align*}
SHP_t(X) &=\bigcup_{Z \in SHP_{t+1}(X)} \{u\in\LdpF{t}:-Z+u\in A_t\}
\\
&=\bigcup_{Z \in SHP_{t+1}(X)} \{u\in\LdpF{t}:-Z+u\in A_{t,t+1}\}
\\
&=\lrcurly{u\in\LdpF{t}:u\in SHP_{t+1}(X)+ A_{t,t+1}}
\\
&=\lrcurly{u \in \LdpF{t}: u\in SHP_{t+1}(X)+\LdpK{p}{t}{K_t}}=SHP_{t+1}(X)\cap\LdpF{t}+\LdpK{p}{t}{K_t},
\end{align*}
for $t\in\{T-1,...,0\}$. Together with $SHP_T(X)=X+K_T$ one obtains a recursive algorithm,
which is shown in \cite{LR11} to be equivalent to a sequence of linear vector optimization problems that can be solved by Benson's algorithm.
\end{remark}

\subsection{Average value at risk}
\label{section_avar}
In this section we will discuss a dynamic set-valued average value at risk with time $t$ parameters $\lambda^t \in \LdiF{t}$, $0 < \lambda^t_i \leq 1$, that is defined by the following dual representation
\begin{equation*}
AV@R_t^{\lambda}(X) := \bigcap_{(\Q,w) \in \W^{\lambda}_{t}} \lrparen{\EQt{-X}{t} + G_t(w)} \cap M_t
\end{equation*}
for any $X \in \LdoF{T}$ and set of eligible portfolios given by a closed subspace $M_t \subseteq \LdoF{t}$. We denote \[\W^{\lambda}_{t} := \lrcurly{(\Q,w) \in \W_t^{\infty}: \diag{w}\lrparen{\diag{\lambda^t}^{-1} \vec{\1} - \diag{\Et{\dQdP}{t}}^{-1} \dQdP} \in \LdiF{T}_+}\] for the regulator version and \[\W^{\lambda}_{t} := \lrcurly{(\Q,w) \in \W_{t,K}^{\infty}: \diag{w}\lrparen{\diag{\lambda^t}^{-1} \vec{\1} - \diag{\Et{\dQdP}{t}}^{-1} \dQdP} \in \LdiF{T}_+}\] for market compatible average value at risk, where $\vec{\1} := \transp{1,...,1} \in \R^d$.
A primal definition of the static set-valued average value at risk can be found in \cite{HRY12}.

\begin{proposition}
\label{propoavar}
$\seq{AV@R^{\lambda}}$ is a normalized closed conditionally coherent dynamic risk measure.
\end{proposition}
\begin{proof}
$\seq{AV@R^{\lambda}}$ is a closed coherent dynamic risk measure by definition, see corollary~\ref{cor_dual}.  It is conditionally coherent by corollary~\ref{cor_conditional_dual} since $\W_t^{\lambda}$ is conditionally conical in the second variable.

To prove normalization, take $u_0 \in AV@R_t^{\lambda}(0)$ and $u_X \in AV@R_t^{\lambda}(X)$. Then, for every $(\Q,w) \in \W^{\lambda}_{t}$ it holds that $u_0 \in G_t(w) \cap M_t$ and $u_X \in \lrparen{\EQt{-X}{t} + G_t(w)} \cap M_t$.  It follows that $u_X + u_0 \in G_t(w) \cap M_t + \lrparen{\EQt{-X}{t} + G_t(w)} \cap M_t \subseteq \lrparen{\EQt{-X}{t} + G_t(w)} \cap M_t$, and therefore $u_X + u_0 \in AV@R_t^{\lambda}(X)$.  The other direction trivially follows from $0 \in AV@R_t^{\lambda}(0)$.
\end{proof}

\begin{remark}
If $\lambda^t = 1$ with $M_t = \LdoF{t}$ then it can be seen that \[\W^{\lambda}_{t} = \lrcurly{(\Q,w) \in \mathcal{M}_d(\P) \times \LdiF{t}_+ \backslash \{0\}: \dQdP \in \LdiF{t}_+}\] and thus $AV@R_t^{\lambda}(X) = \Et{-X}{t} + \LdoF{t}_+$ for any $X \in \LdoF{T}$.  Therefore, for any choice of $M_t$ we have $AV@R_t^{\lambda}(X) = (\Et{-X}{t} + \LdoF{t}_+) \cap M_t$ for any $X \in \LdoF{T}$.
\end{remark}

If $\P(\lambda^t < 1) > 0$ then it can easily be seen that $\seq{AV@R}$ is not recursive since there is no relation between the set of dual variables through time.  Therefore, there is no relation between the acceptance sets between time $t$ and time $t+1$.  Thus by proposition~\ref{propoavar} and theorem~\ref{thm_equiv_tc} the set-valued average value at risk is not a multi-portfolio time consistent risk measure.

By the procedure described in proposition~\ref{lemma_gen_tc} a multi-portfolio time consistent version of the set-valued average value at risk can be obtained, analogous to the scalar version defined in~\cite{CK10}. To obtain a nice dual representation of this composed AV@R further research concerning stability properties are necessary, which are left for further research at that time.

\section*{Acknowledgements}
We would like to thank Patrick Cheridito, Andreas Hamel, Frank Heyde, and Frank Riedel for helpful comments and discussions.
Birgit Rudloff's research was supported by NSF award DMS-1007938 and Zachary Feinstein was supported by the NSF RTG grant 0739195.

\bibliographystyle{rQUF}
\bibliography{biblio}
\end{document}